\documentclass[11pt]{article}
\usepackage{etex,etoolbox}
\usepackage{fullpage}
\usepackage[utf8]{inputenc}
\usepackage[T1]{fontenc}
\usepackage[svgnames]{xcolor}
\usepackage{graphicx}
\usepackage{amsmath, amssymb, amstext}
\usepackage{amsthm}
\usepackage{tikz}
\usepackage{subfig}
\usepackage{mathtools}

\usepackage[hyphens]{url}

\usepackage{pgfplots}
\usepackage{float}

\usepackage[colorlinks]{hyperref}
\definecolor{lightblue}{rgb}{0.5,0.5,1.0}
\definecolor{darkred}{rgb}{0.5,0,0}
\definecolor{darkgreen}{rgb}{0,0.5,0}
\definecolor{darkblue}{rgb}{0,0,0.5}

\hypersetup{colorlinks,linkcolor=darkred,filecolor=darkgreen,urlcolor=darkred,citecolor=darkblue}
\usepackage[titlenumbered, linesnumbered, ruled]{algorithm2e}

\usetikzlibrary{patterns,arrows,decorations.pathreplacing,decorations.pathmorphing,calc}

\makeatletter
\newtheorem*{rep@theorem}{\rep@title}
\newcommand{\newreptheorem}[2]{%
\newenvironment{rep#1}[1]{%
 \def\rep@title{#2 \ref{##1}. (Restated)}%
 \begin{rep@theorem}}%
 {\end{rep@theorem}}}
\makeatother

\theoremstyle{definition}
\newtheorem{theorem}{Theorem}[section]
\newreptheorem{theorem}{Theorem}

\newtheorem{lemma}[theorem]{Lemma}
\newtheorem{definition}[theorem]{Definition}
\newtheorem{corollary}[theorem]{Corollary}

\DeclareMathOperator{\Aut}{Aut}
\DeclareMathOperator{\Sym}{Sym}
\DeclareMathOperator{\CFI}{CFI}

\DeclareMathOperator{\rk}{rk}

\allowdisplaybreaks

\numberwithin{equation}{section}
\numberwithin{figure}{section}

\pgfdeclarelayer{edgelayer}
\pgfdeclarelayer{nodelayer}
\pgfsetlayers{edgelayer,nodelayer,main}
\tikzstyle{normalvertex}=[circle,fill=White,draw=Black]

\title{Benchmark Graphs for Practical Graph Isomorphism}
\author{Daniel Neuen and Pascal Schweitzer\\ 
RWTH Aachen University\\
\texttt{\{neuen,schweitzer\}@informatik.rwth-aachen.de}
}

\newcounter{claimcounter}
\newenvironment{claim}[1][]{
  
  \medskip\par\noindent%
  \ifthenelse{\equal{#1}{}}{%
    \setcounter{claimcounter}{0}\refstepcounter{claimcounter}\textit{Claim~\arabic{claimcounter}.}
  }{%
    \ifthenelse{\equal{#1}{resume}}{%
      \refstepcounter{claimcounter}\textit{Claim~\arabic{claimcounter}.}
    }{%
      \textit{Claim~#1.}
    }
  }
}{
  \par\medskip
}
\newcommand{\uend}{\hfill$\lrcorner$}

 \makeatletter
 \providecommand{\@motsacces}[4]{#4}
 \def\fixstatement#1{%
   \AtEndEnvironment{#1}{%
 
        \xdef\pat@label{\expandafter\expandafter\expandafter
             \@motsacces\csname#1\endcsname\space~\@currentlabel}}}
 \globtoksblk\prooftoks{1000}

\begin{document}

\maketitle

\begin{abstract}
The state-of-the-art solvers for the graph isomorphism problem can readily solve generic instances with tens of thousands of vertices.
Indeed, experiments show that on inputs without particular combinatorial structure the algorithms scale almost linearly.
In fact, it is non-trivial to create challenging instances for such solvers and the number of difficult benchmark graphs available is quite limited. 

We describe a construction to efficiently generate small instances for the graph isomorphism problem that are difficult or even infeasible for said solvers.  

Up to this point the only other available instances posing challenges for isomorphism solvers were certain incidence structures of combinatorial objects (such as projective planes, Hadamard matrices, Latin squares, etc.). 
Experiments show that starting from 1500 vertices our new instances are several orders of magnitude more difficult on comparable input sizes. 
More importantly, our method is generic and efficient in the sense that one can quickly create many isomorphism instances on a desired number of vertices. 
In contrast to this, said combinatorial objects are rare and difficult to generate and with the new construction it is possible to generate an abundance of instances of arbitrary size.

Our construction hinges on the multipedes of Gurevich and Shelah and the Cai-F\"{u}rer-Immerman gadgets that realize a certain abelian automorphism group and have repeatedly played a role in the context of graph isomorphism.
Exploring limits of such constructions, we also explain that there are group theoretic obstructions to generalizing the construction with non-abelian gadgets.
\end{abstract}

\section{Introduction}

The graph isomorphism problem, which asks for structural equivalence of two given input graphs, has been extensively investigated since the beginning of theoretical computer science, both from a theoretical and a practical point of view. Designing isomorphism solvers in practice is a non-trivial task. Nevertheless various efficient algorithms, namely nauty and traces~\cite{mckay}, bliss \cite{bliss}, conauto \cite{conauto}, and saucy \cite{saucy} are available as software packages. These state-of-the-art solvers for the graph isomorphism problem (or more generally graph canonization) can readily solve generic instances with tens of thousands of vertices. Indeed, experiments show that on inputs without particular combinatorial structure the algorithms scale almost linearly.
In practice, this sets the isomorphism problem aside from typical NP-complete problems which usually have an abundance of difficult benchmark instances.
The practical algorithms underlying these solvers differ from the ones employed to obtain theoretical results. Indeed, there is a big disconnect between theory and practice~\cite{DBLP:journals/dagstuhl-reports/BabaiDST15}. One could interpret Babai's recent  breakthrough, the quasipolynomial time algorithm~\cite{DBLP:conf/stoc/Babai16}, as a first step of convergence. The result implies that if graph isomorphism is  NP-complete then all problems in NP have quasi-polynomial time algorithms, which may lead one to also theoretically believe that graph isomorphism is not NP-complete.

In this paper we are interested in creating difficult benchmark instances for the graph isomorphism problem. It would be a misconception to think that for unstructured randomly generated instances graph isomorphism should be hard in practice. Quite the opposite is true. Instances that encompass sufficient randomness usually turn out to be among the easiest instances (see~\cite{nauty:traces:distribution} for extensive tests suites on various kinds of graphs and~\cite{BabErdSelSta:1980,DBLP:conf/focs/Kucera87} for theoretical arguments). In fact, it is non-trivial to create challenging instances for efficient isomorphism solvers and so far the number of difficult benchmark graphs available is quite limited. 

The prime source for difficult instances are graphs arising from combinatorial structures. In 1978, Mathon~\cite{mathon} provided a set of benchmark graphs arising from combinatorial objects such as strongly regular graphs, block designs and coherent configurations. The actual graphs that are provided, having less than 50 vertices, are small by todays standards. Nowadays, larger combinatorial objects yielding difficult graphs are known. The most challenging examples are for instance incidence graphs of projective planes, Hadamard matrices, or Latin squares~(see \cite{nauty:traces:distribution}). It is important to understand that not all such incidence structures yield difficult graphs. Indeed, typically there is an algebraic version that can readily be generated. For example to obtain a projective plane it is possible to consider incidences between one- and two-dimensional subspaces of a three-dimensional vector space over a finite field. If however one uses such an algebraic construction then the resulting graph will automatically have a large automorphism group. The practical isomorphism solvers are tuned to finding such automorphisms and to exploit them for search space contraction. On top of that, there are also theoretical reasons that lead one to believe that graphs with automorphisms make easier examples for the practical tools (see~\cite[Theorem 9]{thesis}). The non-algebraic counterparts of combinatorial structures often do not have the shortcoming of having a large automorphism group, but in contrast to the algebraic constructions they are rare and difficult to generate. 

The second source of difficult instances is based on modifications of the so called Cai-F\"{u}rer-Immerman construction~\cite{cfi}. This construction is connected to the family of  Weisfeiler-Leman algorithms, which for every integer~$k$ contains a~$k$-dimensional version that constitutes a polynomial time algorithm used to distinguish non-isomorphic graphs. The 1-dimensional variant (usually called color refinement) is a subroutine in all competitive practical isomorphism solvers. For larger~$k$, however, the~$k$-dimensional variant becomes impractical due to excessive space consumption. 
For each~$k$, Cai, F\"{u}rer and Immerman construct pairs of non-isomorphic graphs that are not distinguished by the~$k$-dimensional variant of algorithm. In 1996, the construction was
adapted by Miyazaki to explicitly show that the then current version of nauty has exponential running time in the worst case~\cite{miyazaki}. The Cai-F\"{u}rer-Immerman graphs and the Miyazaki graphs constitute families of benchmark graphs that are infeasible for simplistic algorithms for the graph isomorphism problem. Nowadays, however, most state-of-the-art solvers scale reasonably well on these instances (Figures~\ref{fig:benchmark-standard-cfi} and~\ref{fig:benchmark-miyazaki}).

\paragraph{Contribution.} We describe a construction to efficiently generate instances on any desired number of vertices for the graph isomorphism problem that are difficult or even infeasible for all state-of-the-art graph isomorphism solvers.

Experiments show the graphs pose by far the most challenging graphs to date. Already for 1500 vertices our new instances are by several orders of magnitude more difficult than all previously available benchmark graphs on comparable input sizes. More importantly, the algorithm that generates the instances is generic and efficient in the sense that one can quickly create an abundance of isomorphism instances on a desired number of vertices. We can thereby generate instances in size ranges for which no other difficult benchmark graphs are available.

Our construction, the resulting graphs of which we call shrunken multipedes, hinges on a construction of Gurevich and Shelah~\cite{DBLP:journals/jsyml/GurevichS96}. For every~$k$, they describe combinatorial structures that are rigid but cannot be distinguished by the~$k$-dimensional Weisfeiler-Leman algorithm. Guided by the intuition that rigid graphs (i.e., graphs without non-trivial automorphisms) constitute harder instances, we alter the construction to a simple efficient algorithm. In this algorithm we start with a bipartite rigid base graph that is obtained by connecting two explicit graphs (cycles with added diagonals) randomly. We then apply a suitable adaptation of the Cai-F\"{u}rer-Immerman (CFI) construction. We prove that the resulting graph is rigid asymptotically almost surely. Our experiments show that even for small sizes the graphs are already rigid with high probability. The result is a simple randomized algorithm that efficiently creates challenging instances of the graph isomorphism problem. The algorithm is generic in the sense that it can be used to create instances for any desired size range. Moreover it is possible to create many non-isomorphic graphs such that every set of two of them forms a difficult instance.

To create practical benchmark graphs, it is imperative to keep all gadget constructions as small as possible. We therefore employ two techniques to save vertices without changing the local automorphism structure of the gadgets. 
The first technique is based on linear algebra and reduces the vertices in a manner that maintains the rigidity of the graphs. The second technique bypasses certain vertices of the gadgets used in the CFI-construction.

The CFI-gadgets have repeatedly played a role in the context of graph isomorphism. Exploring limits of such constructions, we also explain that there are group theoretic obstructions to generalizing the construction with non-abelian gadgets.

Finally we show with experimental data that our benchmark instances constitute quite challenging problems for all state-of-the-art isomorphism solvers. We compare the running times on the new benchmarks with the combinatorial graphs, the CFI-graphs and the Miyazaki graphs mentioned above.

\paragraph{Theoretical bounds.} 

The benchmark graphs described in this paper are explicitly designed to pose a challenge for practical isomorphism solvers.
Since they have bounded degree they can be solved (theoretically) in polynomial time \cite{DBLP:journals/jcss/Luks82}.
Concerning individualization-refinement algorithms, which the tools mentioned above are, using a different but related construction, exponential lower bounds can be proven.
We refer to~\cite{lower-bound-paper}. 

\paragraph{Obtaining benchmark graphs.}

Our families of benchmark graphs are available for download at~\url{https://www.lii.rwth-aachen.de/research/95-benchmarks.html}\nocite{graphs}. 
However, there are several indications that there is a trade-off between a focus on easy and hard instances.
While some of our instances may constitute the most difficult instances on their respective input size yet, we want to stress the following.
From a practical view point, worst-case running time of isomorphism solvers is not the most important measure.
In many applications, the solvers have to sift through an enormous number of instances, most of which are easy. It is therefore important to perform extremely well on easy instances, and adequately on difficult instances, rather than the other way around.
We refer to the benchmark libraries of nauty and traces (\url{http://pallini.di.uniroma1.it/}), bliss (\url{http://www.tcs.hut.fi/Software/bliss/}) conauto (\url{https://sites.google.com/site/giconauto/}) as well as to the SAT-related benchmarks of saucy (\url{http://vlsicad.eecs.umich.edu/BK/SAUCY/}) for more well-rounded sets of instances in that regard. 
(The first mentioned library explicitly includes the others.)
Nevertheless, we hope that our benchmark graphs help to improve current and future isomorphism solvers.

\section{Preliminaries}
\subsection{Graphs and isomorphisms}

A \emph{graph} is a pair $G=(V,E)$ with vertex set $V = V(G)$ and edge relation $E = E(G)$.
In this work all graphs are simple, undirected graphs. The \emph{neighborhood} of~$v\in V(G)$ is denoted~$N(v)$.
A \emph{path} is a sequence $(v_1,\dots,v_\ell)$ of distinct vertices such that $\{v_i,v_{i+1}\} \in E(G)$ for all $i \in \{1,\dots,\ell-1\}$.
If additionally $\{v_1,v_{\ell}\} \in E(G)$ then the sequence $(v_1,\dots,v_\ell)$ is a \emph{cycle} of $G$.

An \emph{isomorphism} from a graph $G$ to another graph $H$ is a bijective mapping $\varphi\colon V(G) \rightarrow V(H)$ which preserves the edge relation, that is $\{v,w\} \in E(G)$ if and only if $\{\varphi(v),\varphi(w)\} \in E(H)$ for all $v,w \in V(G)$. We use the notation~$G\cong H$ to indicate that such an isomorphism exists. The graph isomorphism problem asks, given two graphs $G$ and $H$, whether there is an isomorphism from $G$ to $H$. 
An \emph{automorphism} of a graph $G$ is an isomorphism from $G$ to itself. By $\Aut(G)$ we denote the automorphisms of $G$.
A graph $G$ is \emph{rigid} (or \emph{asymmetric}) if its automorphism group $\Aut(G)$ is trivial, that is, the only automorphism of $G$ is the identity mapping.

\subsection{The Cai-F\"{u}rer-Immerman Construction}
Our constructions presented in this paper make use of a construction of Cai, F\"{u}rer and Immerman~\cite{cfi}.
In terms of the graph-isomorphism problem, they used the construction to show that for every~$k$ there is a pair of graphs not distinguished by the~$k$-dimensional Weisfeiler-Leman algorithm.
The basic building block for these graphs is the Cai-F\"{u}rer-Immerman gadget $X_3$ depicted in Figure~\ref{fig:cfi-gadget}.
This gadget has~$4$ \emph{inner vertices}~$m_1,\ldots,m_4$ and 3 pairs of~\emph{outer vertices}~$a_i,b_i$ with~$i\in \{1,2,3\}$.
The crucial property of~$X_3$ is that a bijection of the outer vertices mapping the set~$\{a_i,b_i\}$ to itself for all~$i\in \{1,2,3\}$ extends to an automorphism of~$X_3$ if and only if an even number of pairs of outer vertices is swapped.

\begin{figure}
 \centering
 \begin{tikzpicture}
  \node[style=normalvertex,label=right:{$a_1$}] (a1) at (6,2) {};
  \node[style=normalvertex,label=right:{$b_1$}] (b1) at (6,3) {};
  
  \node[style=normalvertex,label=left:{$a_2$}] (a2) at (0,0) {};
  \node[style=normalvertex,label=left:{$b_2$}] (b2) at (0,1) {};
  
  \node[style=normalvertex,label=left:{$a_3$}] (a3) at (0,4) {};
  \node[style=normalvertex,label=left:{$b_3$}] (b3) at (0,5) {};
  
  \node[style=normalvertex,label=above:{$m_1$}] (v1) at (3,4) {};
  \node[style=normalvertex,label=above:{$m_2$}] (v2) at (3,3) {};
  \node[style=normalvertex,label=above:{$m_3$}] (v3) at (3,2) {};
  \node[style=normalvertex,label=above:{$m_4$}] (v4) at (3,1) {};
  
  \path
   (v1) edge (b1)
   (v1) edge (b2)
   (v1) edge (b3)
   (v2) edge (b1)
   (v2) edge (a2)
   (v2) edge (a3)
   (v3) edge (a1)
   (v3) edge (b2)
   (v3) edge (a3)
   (v4) edge (a1)
   (v4) edge (a2)
   (v4) edge (b3);
  
 \end{tikzpicture}
 \caption{Cai-F\"{u}rer-Immerman gadget $X_3$}
 \label{fig:cfi-gadget}
\end{figure}
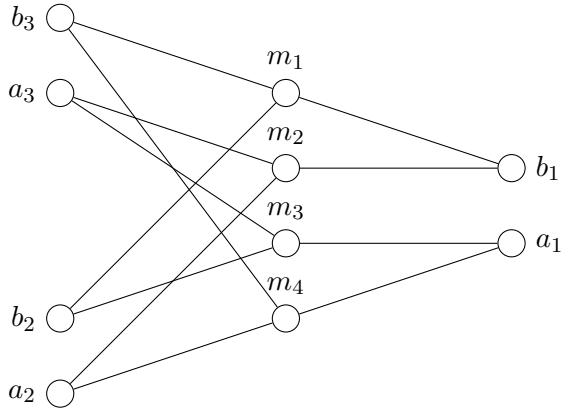

Now let $G$ be a connected 3-regular graph called the \emph{base graph} of the construction and let $T \subseteq E(G)$ be a subset of its edges.
Then the Cai-F\"{u}rer-Immerman construction applied to $G$ results in the following graph $\CFI(G,T)$.
For every vertex $v \in V(G)$ with incident edges $e_i = (v,w_i) \in E(G)$ we add a copy of $X_3$.
The inner vertices are denoted by $m_1(v),\dots,m_4(v)$.
Each pair $\{a_i,b_i\}$ is associated with the edge $e_i$ and we denote the vertices by $a(v,w_i)$ and $b(v,w_i)$.
For every edge $e = \{v,w\} \in E(G)$ with $e \in T$ we add edges $\{a(v,w),b(w,v)\}$ and $\{b(v,w),a(w,v)\}$.
If $e \notin T$ we add edges $\{a(v,w),a(w,v)\}$ and $\{b(v,w),b(w,v)\}$.
The edges $e \in T$ are called \emph{twisted} edges whereas edges $e \notin T$ are \emph{non-twisted}.

Using  the properties of the gadget $X_3$ described above one can now show the isomorphism type of~$\CFI(G,T)$ depends only on the parity of~$|T|$. That is~$\CFI(G,T) \cong \CFI(G,T')$ if and only if $|T| \equiv |T'| \mod 2$. We obtain a pair of non-isomorphic graphs $\CFI(G) := \CFI(G,\emptyset)$ and $\widetilde{\CFI}(G) := \CFI(G,\{e\})$ for some edge $e \in E(G)$.

For appropriate base graphs $G$ (more precisely graphs with sufficiently large tree width~\cite{DBLP:conf/csl/DawarR07}) these graphs are difficult to distinguish by the family of Weisfeiler-Leman algorithms. This also makes them good candidates for being difficult instances for graph isomorphism testing.
However, the practical isomorphism solvers can cope with CFI graphs reasonably well.
The main reason for this is that the underlying algorithms can take advantage of automorphisms of the input graphs to restrict their search space and CFI graphs typically have many automorphisms. 
Indeed, let $C = (v_1,\dots,v_\ell)$ be a cycle in the base graph $G$.
Then there is an automorphism of $\CFI(G)$ that exactly swaps those $a(v,w)$ and $b(v,w)$ for which $\{v,w\}$ is an edge in $C$ and leaves all other $a(v,w)$ and $b(v,w)$ vertices fixed.
Thus, if~$\ell$ is the dimension of the cycle space of~$G$ then~$\CFI(G)$ has at least~$2^{\ell}$ automorphisms.
In the next section we address this issue and describe a similar but probabilistic construction that gives with high probability graphs that have no non-trivial automorphisms.

\section{A Rigid Base Construction}\label{sec-rigid}

\subsection{Desirable properties}

Before we describe our construction we would like to give some intuition about the desirable properties that we want the final graph to have and why we think they are responsible for making the graphs difficult.

\begin{itemize}
\item (rigidity) First of all there are some arguments that may make us believe that rigid graphs are more difficult for isomorphism solvers than graphs with automorphisms.
Indeed, the graphs for which we can prove theoretical lower bounds~\cite{lower-bound-paper} are rigid.
Furthermore, for automorphism group computation, it is possible to obtain upper bounds on the running time in terms of~$|\Gamma|/|\Aut(G)|$ where~$\Gamma$ is a group known to contain all automorphisms (see~\cite{thesis} for a discussion).
However, we can also offer an intuitive explanation.
Isomorphism solvers are used in practice in search algorithms (e.g., SAT-solvers) to exploit symmetry and thereby cut off parts of the search tree.
However, the isomorphism solvers themselves also exploit this strategy in bootstrapping manner, using symmetries to cut off parts of their own search tree.

\item (small constants)
For our practical purposes it is imperative to keep the graphs small. 
We therefore need to diverge from the theoretical construction in~\cite{lower-bound-paper}.
As a crucial part of the construction we devise two methods to reduce the number of vertices while maintaining the difficulty level of the graph.
These reductions are described in the next section.

\item (simplicity) 
We strive to have a simple construction that is not only easy to understand but can also be quickly generated by a simple algorithm.
In contrast, for many other graphs coming from combinatorial constructions it takes far longer to construct the graphs than to perform isomorphism tests.
The simplicity also allows us to construct an abundance of benchmark graphs.

\item (difficult for the WL algorithm) The Weisfeiler-Leman algorithms are a family of theoretical graph isomorphism algorithms.
Rather than giving a description of the algorithm and as to why one may believe it can be used as a measure of the difficulty of a graph, we refer to~\cite{lower-bound-paper}.
\end{itemize}

\subsection{The multipede construction}

In the following we describe an explicit construction for families of similar, rigid, non-isomorphic graphs. This construction is based on the multipedes of Gurevich and Shelah~\cite{DBLP:journals/jsyml/GurevichS96} which yields finite rigid structures and it uses the construction by Cai, F\"{u}rer and Immerman~\cite{cfi}.

Let $G=(V,W,E)$ be a bipartite graph such that every vertex $v \in V$ has degree 3.
Define the multipede construction as follows.
We replace every $w \in W$ by two vertices $a(w)$ and $b(w)$.
For $w \in W$ let $F(w) = \{a(w),b(w)\}$ and for $X \subseteq W$ let $F(X) = \bigcup_{w \in X}F(w)$.
We then replace every $v \in V$ by a copy of $X_3$ and identify the vertices $a_i$ and $b_i$ with $a(w_i)$ and $b(w_i)$ where $w_1,\dots,w_3$ are the neighbors of $v$.
The resulting graph will be denoted by $R(G)$. An example of this construction is shown in Figure~\ref{fig:bipartite:cfi:constr}.
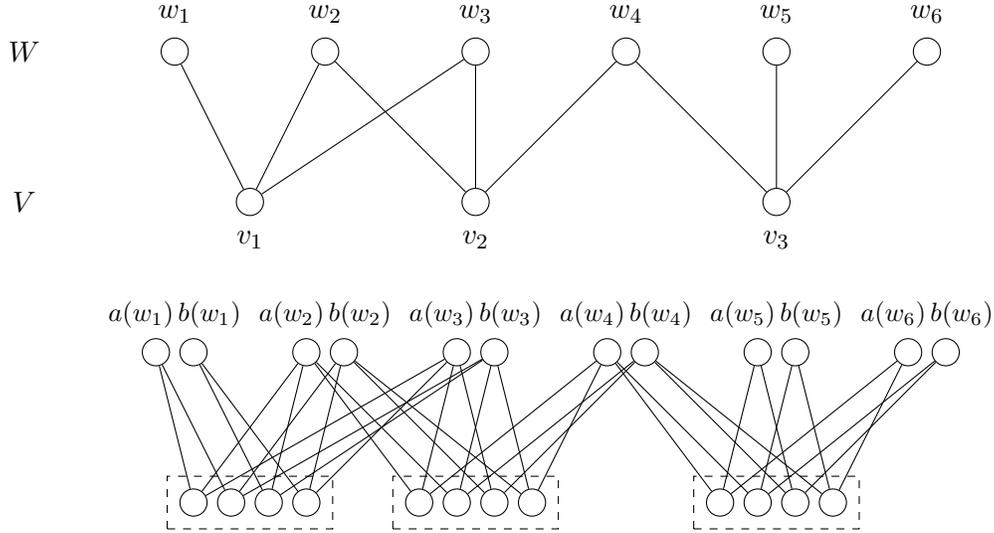
\begin{figure}
\centering
\begin{tikzpicture}
\tikzstyle{normalvertex}=[circle, draw]
		\node [style=normalvertex] (0) at (-11, -2) {};
		\node [style=normalvertex] (1) at (-8, -2) {};
		\node [style=normalvertex] (2) at (-4, -2) {};
		\node [style=normalvertex] (3) at (-12, -0) {};
		\node [style=normalvertex] (4) at (-10, -0) {};
		\node [style=normalvertex] (5) at (-8, -0) {};
		\node [style=normalvertex] (6) at (-6, -0) {};
		\node [style=normalvertex] (7) at (-2, -0) {};
		\node [style=normalvertex] (8) at (-4, -0) {};
		\node [style=normalvertex] (9) at (-11.75, -4) {};
		\node [style=normalvertex] (10) at (-12.25, -4) {};
		\node [style=normalvertex] (11) at (-10.25, -4) {};
		\node [style=normalvertex] (12) at (-9.75, -4) {};
		\node [style=normalvertex] (13) at (-8.25, -4) {};
		\node [style=normalvertex] (14) at (-7.75, -4) {};
		\node [style=normalvertex] (15) at (-11.25, -6) {};
		\node [style=normalvertex] (16) at (-11.75, -6) {};
		\node [style=normalvertex] (17) at (-10.75, -6) {};
		\node [style=normalvertex] (18) at (-10.25, -6) {};
		\node [style=normalvertex] (19) at (-8.75, -6) {};
		\node [style=normalvertex] (20) at (-8.25, -6) {};
		\node [style=normalvertex] (21) at (-7.75, -6) {};
		\node [style=normalvertex] (22) at (-7.25, -6) {};
		\node [style=normalvertex] (23) at (-6.25, -4) {};
		\node [style=normalvertex] (24) at (-5.75, -4) {};
		\node [style=normalvertex] (25) at (-4.25, -4) {};
		\node [style=normalvertex] (26) at (-3.75, -4) {};
		\node [style=normalvertex] (27) at (-2.25, -4) {};
		\node [style=normalvertex] (28) at (-1.75, -4) {};
		\node [style=normalvertex] (29) at (-4.75, -6) {};
		\node [style=normalvertex] (30) at (-4.25, -6) {};
		\node [style=normalvertex] (31) at (-3.75, -6) {};
		\node [style=normalvertex] (32) at (-3.25, -6) {};
	
		\draw (3) to (0);
		\draw (0) to (4);
		\draw (0) to (5);
		\draw (4) to (1);
		\draw (1) to (5);
		\draw (1) to (6);
		\draw (6) to (2);
		\draw (2) to (7);
		\draw (8) to (2);
		\draw (10) to (16);
		\draw (10) to (15);
		\draw (9) to (17);
		\draw (9) to (18);
		\draw (11) to (16);
		\draw (11) to (17);
		\draw (12) to (15);
		\draw (12) to (18);
		\draw (13) to (16);
		\draw (13) to (18);
		\draw (14) to (15);
		\draw (17) to (14);
		\draw (19) to (11);
		\draw (11) to (20);
		\draw (21) to (12);
		\draw (12) to (22);
		\draw (13) to (19);
		\draw (13) to (21);
		\draw (14) to (20);
		\draw (14) to (22);
		\draw (23) to (19);
		\draw (23) to (22);
		\draw (24) to (20);
		\draw (24) to (21);
		\draw (23) to (29);
		\draw (23) to (30);
		\draw (24) to (31);
		\draw (24) to (32);
		\draw (25) to (29);
		\draw (25) to (31);
		\draw (26) to (30);
		\draw (26) to (32);
		\draw (27) to (29);
		\draw (27) to (32);
		\draw (28) to (30);
		\draw (28) to (31);
	
\node at (-14,0) {$W$};

\node at (-14,-2) {$V$};

\node at (-11,-2.5) {$v_1$};

\node at (-8,-2.5) {$v_2$};

\node at (-4,-2.5) {$v_3$};

\foreach \x in {1,...,6}{
\node at (-14+2*\x,0.5) {$w_\x$};
}

\foreach \x in {0,3,7}{
\draw[dashed]  (-12.1+\x,-5.65) rectangle (-9.9+\x,-6.35);
}
\foreach \x in {1,...,6}{

\node at (-14+2*\x,-3.5) {\small{$a(w_\x)\, b(w_\x)$}};
}
\end{tikzpicture}
 \caption[The bipartite~CFI-construction]{The figure depicts a base graph~$G$ on the top and the graph~$R(G)$ obtained by applying the multipede construction on the bottom.}
 \label{fig:bipartite:cfi:constr}
\end{figure}

\begin{definition}
 Let $G=(V,W,E)$ be a bipartite graph such that every vertex $v \in V$ has degree~3.
 We say $G$ is \emph{odd} if for every $\emptyset \neq X \subseteq W$ there is some $v \in V$ such that $|X \cap N(v)|$ is odd.
\end{definition}

If $G$ is not odd we call it \emph{even} and the corresponding set $X$ providing the proof for not being odd is called a \emph{witness}.

For a graph $G$ and a vertex $v \in V(G)$ we define the \emph{second neighborhood} of~$v$ to be the set $N_G^{2}(v) = \{u \in V(G) \mid u \neq v \wedge \exists w \in V(G) \colon \{v,w\}, \{w,u\} \in E(G) \}$.

\begin{lemma}
 \label{la:rigidity-cfi-graph}
 Let $G = (V,W,E)$ be a bipartite graph, such that
 \begin{enumerate}
  \item $|N(v)| = 3$ for all $v \in V$,
  \item $G$ is odd,
  \item $G$ is rigid, and
  \item there are no distinct $w_1,w_2 \in W$, such that $N_G^{2}(w_1) = N_G^{2}(w_2)$.
 \end{enumerate}
 Then $R(G)$ is rigid.
\end{lemma}

\begin{proof}
 Suppose $\psi \in \Aut(R(G))$.
 First observe that $F(W)$ is invariant under $\psi$ since every vertex in $F(W)$ has even degree whereas all other vertices have degree $3$.
 For $w \in W$ it holds that \[N_{R(G)}^{2}(a(w)) = N_{R(G)}^{2}(b(w)) = F(N_G^{2}(w)).\]
 Let $w_1,w_2 \in W$, such that $\psi(a(w)) \in F(w_1)$ and $\psi(b(w)) \in F(w_2)$.
 Then $F(N_G^{2}(w_1)) = F(N_G^{2}(w_2))$ and consequently $N_G^{2}(w_1) = N_G^{2}(w_2)$.
 By the last condition we conclude that $w_1 = w_2$.
 So for every $w \in W$ there is some $w' \in W$, such that $\psi(F(w)) = F(w')$.
 From this it easily follows that the same holds for vertices in $V$, that is for every $v \in V$ there is some $v' \in V$, such that $\{\psi(m_1(v)),\dots,\psi(m_4(v))\} = \{m_1(v'),\dots,m_4(v')\}$.
 So $\psi$ induces a permutation $\psi' \in \Sym(V(G))$ with $\psi'(w) = w'$ for all $w \in W$ and $\psi'(v) = v'$ for all $v \in V$.
 Since $\psi'$ is an automorphism of $G$ we conclude that $\psi'$ is the identity.
 So for every $w \in W$ it holds that $\psi(F(w)) = F(w)$.
 
 Now let $X = \{w \in W \mid \psi(a(w)) = b(w)\}$ and suppose $X \neq \emptyset$.
 Since $G$ is odd there is some $v \in V(G)$, such that $|X \cap N(v)|$ is odd.
 But then $\psi$ restricts to an automorphism of $X_3$ that twists an odd number of pairs $\{a_i,b_i\}$ contradicting the properties of $X_3$.
 Thus $X = \emptyset$ and every vertex in $F(W)$ is fixed under $\psi$.
 But this immediately implies that $\psi$ is the identity.
\end{proof}

In the following we will present a simple randomized algorithm to construct rigid bipartite graphs that are odd.
Let $G$ be a 3-regular graph and let $\sigma: E(G) \rightarrow E(G)$ be a random permutation of the edges of $G$.
We define the bipartite graph $B(G,\sigma) = (V_B, W_B, E_B)$ by setting
\begin{align*}
 &V_B = V(G) \times \{0,1\},\\
 &W_B = E(G),\\
 &E_B = \{\{(v,0),e\}\mid v \in e\} \cup \{\{(v,1),e\}\mid v \in \sigma(e)\}.
\end{align*}

Thus, the edges of~$G$ correspond to the vertices in the partition class~$W_B$ of~$B(G,\sigma)$.  Each such edge~$e  = \{v,w\}$ has an associated edge~$\sigma(e) = \{v',w'\}$ and~$e$ has exactly four neighbors, namely~$(v,0)$,~$(w,0)$, $(v',1)$ and~$(w',1)$. Another way of visualizing this construction is to start with two copies of~$G$, subdivide all edges in each copy and then identify the newly added vertices in the one copy with the newly added vertices in the other copy using a randomly chosen matching.

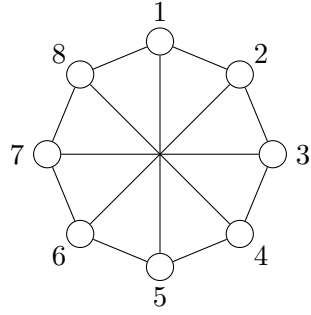
\begin{figure}
 \centering
 \begin{tikzpicture}
  \foreach \pos in {1,...,8}{
   \node[style=normalvertex] (\pos) at (45*\pos:1.5) {};
   \node at (135-45*\pos:1.9) {$\pos$};
  }
  
  \path
   (1) edge (2)
   (2) edge (3)
   (3) edge (4)
   (4) edge (5)
   (5) edge (6)
   (6) edge (7)
   (7) edge (8)
   (8) edge (1)
   (1) edge (5)
   (2) edge (6)
   (3) edge (7)
   (4) edge (8);
 \end{tikzpicture}
 \caption{The graph $G_4$}
 \label{fig:graph-g4}
\end{figure}

We define~$G_n$ to be the~\emph{$2n$-cycle with diagonals added}. More precisely, $G_n := (V_n,E_n)$ where $V_n = \{1,\dots,2n\}$ and $E_n = \{\{i,i+1 \bmod 2n\}\mid 1 \leq i \leq 2n\} \cup \{\{i,i+n\}\mid 1 \leq i \leq n\}$. See Figure~\ref{fig:graph-g4} for an example. We call the edges of the form~$\{i,i+n\}$ \emph{diagonals}.

In reference to~\cite{DBLP:journals/jsyml/GurevichS96}  we call the graphs~$R(B(G_n,\sigma))$ the multipede graphs. The entire algorithm to construct $R(B(G_n, \sigma))$ is again summarized by Algorithms \ref{alg:bipartite-construction} and \ref{alg:rigid-cfi}. It can be computed that the graph~$R(B(G_n,\sigma))$ has an average degree of~$48/11\approx 4.363$.

\begin{algorithm}
 \caption{Rigid Odd Base Graph Construction}
 \label{alg:bipartite-construction}
 \DontPrintSemicolon
 \SetKwInOut{Input}{Input}
 \SetKwInOut{Output}{Output}
 \Input{$n \in \mathbb{N}$}
 \Output{Graph $B = B(G_n,\sigma)$}
 \BlankLine
 $V(B) := V(G_n) \times \{0,1\} \cup E(G_n)$\;
 $E(B) := \emptyset$\;
 randomly choose permutation $\sigma \in \Sym(E(G_n))$\;
 \For{$v \in V(G_n) \times \{0\}$}{
   $E(B) := E(B) \cup \{ve \mid e \in E(G_n) \text{ with } v \in e\}$\;
 }
 \For{$v \in V(G_n) \times \{1\}$}{
   $E(B) := E(B) \cup \{ve \mid e \in E(G_n) \text{ with } v \in \sigma(e)\}$\;
 }
 return $B$
\end{algorithm}

\begin{algorithm}
 \caption{The multipede graph construction}
 \label{alg:rigid-cfi}
 \DontPrintSemicolon
 \SetKwInOut{Input}{Input}
 \SetKwInOut{Output}{Output}
 \Input{Bipartite graph $B=(V,W,E)$ such that every vertex $v \in V$ has degree $3$}
 \Output{Graph $G = R(B)$}
 \BlankLine
 $V(G) := \{m_i(v)\mid v \in V, i \in \{0,1\}^{3} \text { with number of ones in~$i$ is even}\} \cup \{a(w),b(w)\mid w \in W\}$\;
 $E(G) := \emptyset$\;
 \For{$v \in V$}{
  $\{w_1,w_2,w_3\} := N(v)$ \tcc*[r]{fix order on the neighbors of $v$}
  \For{$j \in \{1,2,3\} \text{ and } i \in \{0,1\}^{3}  \text{ with number of ones in~$i$ is even } $}{
   \eIf{$i_j = 0$}{
    $E(G) := E(G) \cup \{\{m_i(v),a(w_j)\}\}$ \tcc*[r]{connect $m_i(v)$ to $a(w_j)$}
   }{
    $E(G) := E(G) \cup \{\{m_i(v),b(w_j)\}\}$ \tcc*[r]{connect $m_i(v)$ to $b(w_j)$}
   }
  }
 }
 return $G$
\end{algorithm}

We can analyze this construction and show that for the base graphs $G_n$ the resulting graphs $R(B(G_n,\sigma))$ are with high probability rigid. For this, 
building on Lemma \ref{la:rigidity-cfi-graph}, we show that with high probability the graph $B(G_n,\sigma)$ is odd, rigid, and has no distinct vertices in $W_B$ with equal second neighborhoods.

\begin{theorem}\label{thm:construction:is:rigid}
 The probability that $R(B(G_n, \sigma))$ is not rigid is in $\mathcal{O}\left(\frac{\log^{2}n}{n}\right)$.
\end{theorem}

The lengthy proof of the theorem can be found in Appendix~\ref{app:proof:of:rigidity:theorem}.

While it is difficult to extract hidden constants in the theorem, our experiments show that even for small $n$ the probability that~$R(B(G_n, \sigma))$ is not rigid is close to~$0$.
Let us also remark that we can create non-isomorphic graphs by switching the neighborhoods of the vertices~$a(w)$ and~$b(w)$ within one of the gadgets when creating~$R(B(G_n,\sigma))$ from~$B(G_n,\sigma)$.
Since the graph is rigid, doing so for different subsets~$W'\subseteq W$ we 
obtain many graphs that pose difficult isomorphism instances.

\section{The shrunken multipedes}
\label{sec:reduce-vertices}

As we described before we are interested in keeping the number of vertices of our construction small.
In this section we describe two methods to reduce the number of vertices of the multipede graphs while, according to our experiments, essentially preserving the difficulty.

\subsection{A linear algebra reduction}\label{subsec:reduce-linalg}

For a bipartite graph $G = (V,W,E)$ let $A_G \in \mathbb{F}_2^{V \times W}$ be the matrix with $A_{v,w} = 1$ if and only if $vw \in E(G)$ and let $\rk(A)$ be the $\mathbb{F}_2$-rank of $A$.

\begin{lemma}
 Let $G=(V,W,E)$ be a bipartite graph.
 Then $G$ is odd if and only if $\rk(A_G) = |W|$.
\end{lemma}

\begin{proof}
 For the forward direction suppose $x \in \mathbb{F}_2^{W}$ such that $A_Gx = 0$. Set $X = \{w \in W \mid x_w = 1\}$.
 Suppose towards a contraction that $X \neq \emptyset$.
 Since $G$ is odd there is some $v \in V$ such that $|N(v) \cap X|$ is odd.
 But then $(A_G)_vx = 1$ where $(A_G)_v$ is the $v$-th row of $A_G$.
 This is a contraction and thus, $\{x\in \mathbb{F}_2^{W} \mid A_Gx = 0\} = \{0\}$.
 So $\rk(A_G) = |W|$.
 
 Suppose $\rk(A_G) = |W|$.
 Then $\{x\in \mathbb{F}_2^{W} \mid A_Gx = 0\} = \{0\}$.
 Let $\emptyset \neq X \subseteq W$ and let $x \in \mathbb{F}_2^{W}$ be the vector with $x_w = 1$ if and only if $w \in X$.
 Since $x \notin \{x\in \mathbb{F}_2^{W} \mid A_Gx = 0\}$ there is some $v \in V$ such that $(A_G)_vx = 1$.
 But this means that $|N(v) \cap X|$ is odd.
\end{proof}

\begin{corollary}
 \label{cor:reduce-size-odd}
 Let $G=(V,W,E)$ be an odd bipartite graph.
 Then there is some $V' \subseteq V$ with $|V'| \leq |W|$ such that the induced subgraph $G[V' \cup W]$ is odd.
\end{corollary}

Moreover, using Gaussian elimination, we can compute such a set in polynomial time.
Now suppose $B(G_n,\sigma)$ is odd.
Then we can use the previous corollary to compute an induced subgraph $B^{*}\coloneqq B^{*}(G_n,\sigma)$ of $B(G_n,\sigma)$ which is odd and has fewer vertices.
Applying the rigid base construction to $B^{*}$ the resulting graph has $|V(R(B^{*}))| = 4\cdot 3 \cdot n + 2\cdot 3 \cdot n = 18n$  vertices. The average degree has decreased to~4.
In comparison, $|V(R(B(G_n,\sigma)))| = 4\cdot 4 \cdot n + 2\cdot 3 \cdot n = 22n$.

\subsection{Bypassing the outer vertices}\label{subsec:reduce-outer}

Consider the CFI-construction when applied to a 3-regular base graph on~$n$ vertices.
An easy trick is to identify the $a$ and $b$ vertices for each edge of the base graph instead of connecting them by an edge.
More precisely, let $e = \{v,w\}$ be an edge in the base graph.
Then we can identify $a(v,w)$ with $a(w,v)$ and $b(v,w)$ with $b(w,v)$ in case of a non-twisted connection or identify $a(v,w)$ with $b(w,v)$ and $b(v,w)$ with $a(w,v)$ in case of a twisted connection.
This way the number of vertices can be reduced from $10n$ to $7n$. 
This can further be improved by removing all $a$ and $b$ vertices altogether and connecting inner vertices $m_i(v)$ to $m_j(w)$ if both are connected to either $a(v,w)$ or $b(v,w)$.
In this form the construction only requires $4n$ vertices.

A similar technique can also be applied to the rigid base construction. For this we bypass all vertices of degree~$8$ by directly joining their neighbors and then removing all such vertices.
Given a bipartite base graph~$G$ we denote by~$R^*(G)$ the graph obtained form~$R(G)$ by bypassing the outer vertices.
When reducing the graphs in that fashion one has to be aware that combinatorial structure of the graph might get lost.
In particular it can be the case that vertices of degree~$4$ (in the base graph) had the same neighborhood before they were removed, leading to inhomogeneity in the graph $R^*(G)$.
In fact our experiments show that there is wider spread concerning the difficulty of the graphs $R^*(G)$ than the graphs $R(G)$.
In other words, the rigidity and difficulty of those graphs depends more heavily on the choice of the matching~$\sigma$ than for the rigid base construction. However, the graphs $R^*(G)$ turn out to be more difficult than those obtained from the multipede construction (see Section~\ref{sec:benchmark}).

\subsection{Applying both reductions}

We can combine the two vertex reduction techniques presented in this section by first applying Corollary \ref{cor:reduce-size-odd} to the base graph and then bypassing the outer vertices.
The number of vertices in the combined reduction decreases to $12 n$. It is not very difficult to see that the graphs have an average degree of at most~24, but the average degree varies among graphs on the same number of vertices since there may be multiple bypass edges that would have the same endpoints. 

Our experiments show that the combination of the two reductions (Section~\ref{sec:benchmark}) yields the most difficult graphs.
We call the resulting graphs $R^*(B^*(G_n,\sigma))$ the \emph{shrunken multipedes}.

\section{Cai-F\"{u}rer-Immerman gadgets for other groups}\label{sec:other:groups}

The constructions described in the previous sections revolve around the particular CFI-gadget depicted in Figure~\ref{fig:cfi-gadget}.
On the outer vertices, the automorphism group of that CFI-gadget induces the set of all permutations that swap an even number of pairs.
This corresponds to the subgroup of~$(\mathbb{Z}_2)^3$ of elements~$(g_1,g_2,g_3)$ with~$g_1 g_2 g_3 = 1$.
A natural idea to push the CFI-construction to its limits is to encode other groups by the use of other gadgets. 

Indeed, it is not difficult to see that for every permutation group~$\Delta$ acting on a set~$\Omega$ there is a graph gadget~$X_{\Delta}$ with a vertex set containing~$\Omega$ as $\Aut(X_{\Delta})$-invariant subset such that~$\Aut(X_{\Delta})\cong \Delta$ and~$\Aut(G)|_{\Omega} = \Delta$. (For a description of such a construction see for example \cite[Lemma 16]{DBLP:conf/swat/OtachiS14}.) In other words, every permutation group~$\Delta$ can be \emph{realized} by a graph gadget~$X_{\Delta}$.
As explained before, the original CFI-gadget realizes a group that is a subdirect product of~$\mathbb{Z}_2^3$, which is in particular an abelian group.
In analogy to our previous terminology we call the vertices in~$\Omega$ the outer vertices and the vertices in~$V(X_\Delta)\setminus \Omega$ the inner vertices.
Note that for our purpose of creating benchmark graphs, we are interested in keeping the number of inner vertices small. 

In our constructions,~$\Omega$ naturally decomposes into sets~$C_1,C_2,C_3,\ldots ,C_t$ of equal size~$n$ as the classes of the outer vertices.
We think of these sets as colored with different colors and as usual we will only consider color preserving isomorphisms, that is permutations that map each vertex to a vertex of the same color. If we insert a gadget~$X_{\Delta}$ with~$\Omega = C_1\cup C_2\cup \dots C_t$ we obtain a subgroup of~$\Sym(C_1)\times \Sym(C_2)\times \dots\times \Sym(C_t)$, the direct product of symmetric groups.

For simplicity, in the rest of this section we focus on the case~$t=3$. Thus we consider colored graph gadgets~$X_{\Delta}$ such that~$C_1\cup C_2\cup C_3\subseteq V(X_{\Delta})$ and such that the automorphism group~$\Aut(X_{\Delta})$ of~$X_{\Delta}$ induces on~$C_1\cup C_2\cup C_3$ a certain subgroup~$\Delta \leq \Sym(C_1)\times \Sym(C_2)\times \Sym(C_3)$.

Once we have a suitable gadget~$X_{\Delta}$ we can use it to obtain a generalized CFI-construction that we describe next.
For a~$3$-regular base graph~$G$, we define the graph $X(G, \Delta, \gamma)$ with~$\gamma\in \Sym(\{1,\ldots,n\})$ as follows.
Each vertex $v$ of $G$ is replaced by a copy of $X_{\Delta}$ where the $a_{i,j}$, $1 \leq j \leq n$, correspond to one of the three incident edges.
For an edge $(v,w)$ in $G$ we denote these vertices by $a(v,w)_j$, $1 \leq j \leq n$.
Finally, we connect vertices $a(v,w)_j$ and $a(w,v)_j$ for all $1 \leq j \leq n$ and all but one edges $\{v,w\}$ in $G$ whereas for the last edge we connect $a(v,w)_j$ and $a(w,v)_{\gamma(j)}$ for $1 \leq j \leq n$.

This brings us to the question, what types of gadgets are suitable for our purpose of creating hard benchmark graphs.

Various results in the context of algorithmic group theory can lead one to believe that small groups and also abelian groups may be easier to handle algorithmically than large or non-abelian groups (\cite{DBLP:conf/stacs/BabaiQ12,DBLP:journals/jcss/Luks82,DBLP:conf/mfcs/ZaidGGP14}).
Thus, one may wonder whether more difficult benchmark graphs arise when employing graph gadgets inducing more complicated abelian or even non-abelian automorphism groups.
We now explore these thoughts. We start with several examples.

\subsection{Examples for more general groups}
\label{sec:other-groups-constructions}

\paragraph{General abelian groups.}
Let $\Gamma \leq S_n$ be an abelian permutation group on~$n$ elements.
Then~$\Delta \coloneqq \{(a,b,c) \in \Gamma^{3}\mid abc = 1\}$ is a subgroup of~$S_n \times S_n\times S_n$.
Here~$S_n$ denotes the symmetric group of the set~$\{1,\ldots,n\}$.

We can describe an explicit gadget $X_\Delta$ that induces this group~$\Delta$ for the case in which~$\Gamma$ is transitive.
For~$i\in \{1,2,3\}$ we define~$C_i = \{1,\ldots,n\}$ and let the disjoint union~$C_1\cup C_2 \cup C_3$ be the outer vertices of the gadget.
For every~$(g_1,g_2,g_3) = g\in \Delta$ we have an inner vertex~$m_g$. Then the set of inner vertices is~$M = \{m_g\mid g\in \Delta\}$.
For~$a_i \in C_i$ and~$m_{(g_1,g_2,g_3)}\in M$ we add an edge~$\{a_i,m_{(g_1,g_2,g_3)}\}$ if~$g_i$ maps the element~$1$ to~$a_i$, (i.e.,~$g_i(1) = a_i$.)
For $\Gamma = \mathbb{Z}_2$ we recover the CFI-gadget. The gadget for the group~$\Gamma = \mathbb{Z}_3$ is shown in Figure~\ref{fig:z3:gadget}.
Let us point out that for $\Gamma = \mathbb{Z}_k$ these gadgets were used in \cite{gi-modkl-hard} to show hardness of graph isomorphism for certain complexity classes.

\tikzstyle{normalvertex}=[circle,draw]

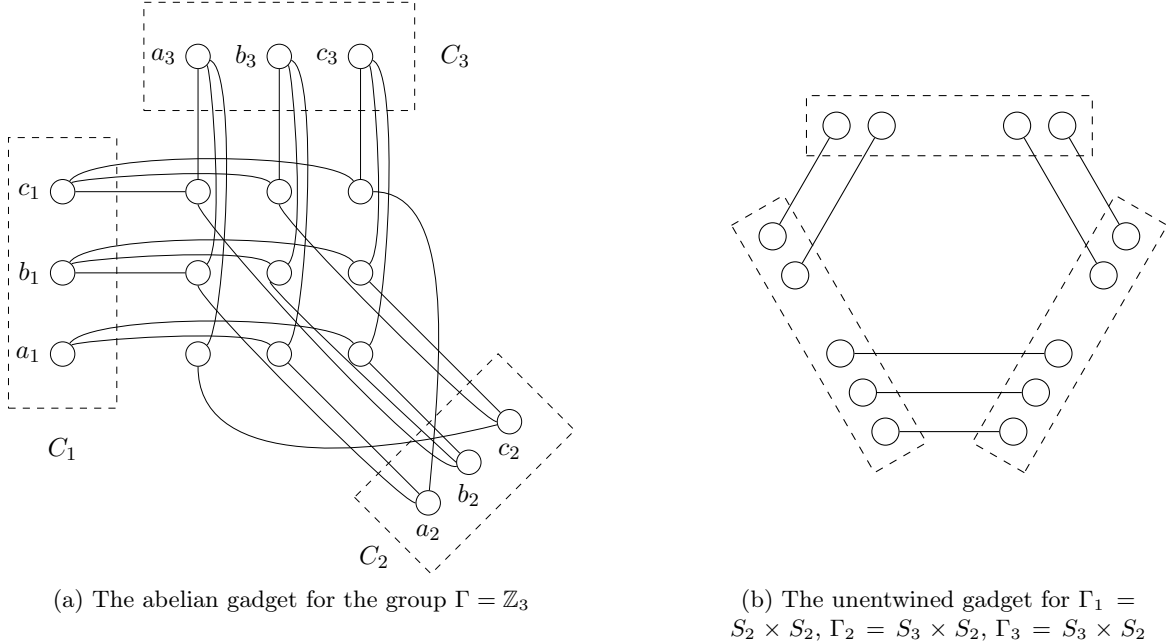
\begin{figure}
 \centering
 \captionsetup[subfloat]{ justification=Centering}
 \subfloat[The abelian gadget for the group $\Gamma = \mathbb{Z}_3$\label{fig:z3:gadget}]{\scalebox{0.9}{ 
 \begin{tikzpicture}[scale =  0.8]
  \node[style=normalvertex,label=left:{$a_1$}] (a1) at (-1.5,3) {};
  \node[style=normalvertex,label=left:{$b_1$}] (b1) at (-1.5,4.5) {};
  \node[style=normalvertex,label=left:{$c_1$}] (c1) at (-1.5,6) {};
  
  \node[style=normalvertex,label=below:{$a_2$}] (a2) at (5.25,0.25) {};
  \node[style=normalvertex,label=below:{$b_2$}] (b2) at (6,1) {};
  \node[style=normalvertex,label=below:{$c_2$}] (c2) at (6.75,1.75) {};
  
  \node[style=normalvertex,label=left:{$a_3$}] (a3) at (1,8.5) {};
  \node[style=normalvertex,label=left:{$b_3$}] (b3) at (2.5,8.5) {};
  \node[style=normalvertex,label=left:{$c_3$}] (c3) at (4,8.5) {};
  
  \node[style=normalvertex,label=below left:{}] (v1) at (1,6) {};
  \node[style=normalvertex,label=below left:{}] (v2) at (2.5,6) {};
  \node[style=normalvertex,label=below left:{}] (v3) at (4,6) {};
  
  \node[style=normalvertex,label=below left:{}] (v4) at (1,4.5) {};
  \node[style=normalvertex,label=below left:{}] (v5) at (2.5,4.5) {};
  \node[style=normalvertex,label=below left:{}] (v6) at (4,4.5) {};
  \node[style=normalvertex,label=below left:{}] (v7) at (1,3) {};
  \node[style=normalvertex,label=below left:{}] (v8) at (2.5,3) {};
  \node[style=normalvertex,label=below left:{}] (v9) at (4,3) {};
  
  \draw  (a3) edge (v1);
  \draw  (b3) edge (v2);
  \draw  (c3) edge (v3);
  \draw  (c1) edge (v1);
  \draw  (b1) edge (v4);
  
  \draw (a3) .. controls (1.25,8.25) and (1.5,5) .. (v4);
  \draw (a3) .. controls (1.75,8.25) and (1.5,3.5) .. (v7);
  \draw (b3) .. controls (2.75,8.25) and (3,5) .. (v5);
  \draw (b3) .. controls (3.25,8.25) and (3,3.5) .. (v8);
  \draw (c3) .. controls (4.25,8.25) and (4.5,5) .. (v6);
  \draw (c3) .. controls (4.75,8.25) and (4.5,3.5) .. (v9);
  
  \draw (c1) .. controls (-1.25,6.25) and (2,6.5) .. (v2);
  \draw (c1) .. controls (-1,6.75) and (3.5,6.75) .. (v3);
  \draw (b1) .. controls (-1.25,4.75) and (2,5) .. (v5);
  \draw (b1) .. controls (-1,5.25) and (3.5,5.25) .. (v6);
  \draw (a1) .. controls (-1.25,3.25) and (2,3.5) .. (v8);
  \draw (a1) .. controls (-1,3.75) and (3.5,3.75) .. (v9);
  
  \draw  (c2) edge (v6);1
  \draw  (b2) edge (v9);
  \draw  (a2) edge (v8);
  
  \draw (v5) .. controls (2.25,4.25) and (5.75,1) .. (b2);
  \draw (v1) .. controls (1,5.5) and (5.25,0.75) .. (b2);
  
  \draw (v4) .. controls (1,4) and (4.75,0.25) .. (a2);
  \draw (v2) .. controls (2.5,5.5) and (6.25,1.75) .. (c2);
  \draw (v3) .. controls (5.75,6) and (5.5,3) .. (a2);
  \draw (v7) .. controls (1,1) and (3.75,1) .. (c2);
  
  \draw[dashed]  (-2.5,7) rectangle (-0.5,2);
  \node at (-1.5,1.25) {$C_1$};
  
  \draw[dashed]  (0,9.5) rectangle (5,7.5);
  \node at (5.75,8.5) {$C_3$};
  
  \draw[dashed, rotate = 45]  (3,-2.5) rectangle (6.75,-4.5);
  \node at (4.25,-0.75) {$C_2$};
 \end{tikzpicture}
 }}
 \quad \quad \quad \quad
 \captionsetup[subfloat]{justification=Centering}
 \subfloat[The unentwined gadget for $\Gamma_1 = S_2 \times S_2$, $\Gamma_2 = S_3 \times S_2$, $\Gamma_3 = S_3 \times S_2$\label{no:3:dependence:gadget}]{
 \begin{tikzpicture}
  \node[style=normalvertex] (u1) at (0:1.7) {};
  \node[style=normalvertex] (u2) at (0:2.3) {};
  
  \node[style=normalvertex] (u3) at (0:4.1) {};
  \node[style=normalvertex] (u4) at (0:4.7) {};
  
  \foreach \x in {1,2,4,5,6}{
   \node[style=normalvertex] (v\x) at (300:1.1+0.6*\x) {};
   \node[style=normalvertex] (w\x) at ($(6.4,0)+(240:1.1+0.6*\x)$) {};
  }
  
  \draw (u1) to (v1);
  \draw (u2) to (v2);
  \draw (u3) to (w2);
  \draw (u4) to (w1);
  
  \draw (v4) to (w4);
  \draw (v5) to (w5);
  \draw (v6) to (w6);
  
  \draw[dashed] (1.3,-0.4) rectangle (5.1,0.4);
  \draw[dashed,shift={(300:1.3)},rotate=300] (0,-0.4) rectangle (3.8,0.4);
  \draw[dashed,shift={($(6.4,0)+(240:5.1)$)},rotate=60] (0,-0.4) rectangle (3.8,0.4);
  \begin{scope}
   \path[use as bounding box] (0,0) rectangle (2,-6);
  \end{scope}
 \end{tikzpicture}
 }
 \caption{Gadgets for groups other than~$\mathbb{Z}_2$.}
\end{figure}

Applying the gadgets to suitable base graphs, this gives us a generalization of the rigid base construction to arbitrary abelian groups. 
However the difficulty of the graphs does not increase for the algorithms tested here as we will see in Section \ref{sec:benchmark}.

Since in various contexts abelian groups are algorithmically easier to handle  than non-abelian groups, we would like to generalize the rigid base construction also for non-abelian groups.

\paragraph{Semidirect products.}
Our next example is that of a semidirect product. Let~$\Gamma = N \rtimes H\leq S_n$ be a semidirect product with an abelian normal subgroup $N$.
For example~$\Gamma$ could be the dihedral group~$D_{n}$.
Then~$\Delta \coloneqq \{(ah, bh, ch) \in \Gamma^{3}\mid a,b,c \in N,\,h \in H,\,abc = 1\}$ is a subgroup of~$S_n\times S_n\times S_n$.

Analogous to the abelian case we can build a gadget $X_\Delta$ realizing the group~$\Delta$.
For~$n\geq 3$, we describe an explicit construction for the case in which~$\Gamma$ is the dihedral group~$D_{n}$ of order~$2n$ (i.e. the automorphism group of an~$n$-cycle).

As before, we set~$C_i = \{1,\ldots,n\}$ and the disjoint union~$C_1\cup C_2\cup C_3$ forms the set of outer vertices.
For every~$(g_1,g_2,g_3) = g\in \Delta$ we have an inner vertex~$m_g$. Then the set of inner vertices is~$M = \{m_g\mid g\in \Delta\}$.
For~$a_i \in C_i$ and~$m_{(g_1,g_2,g_3)}\in M$ we add an edge~$\{a_i,m_{(g_1,g_2,g_3)}\}$ if~$g_i$ maps one of the elements in~$\{1,2\}$ to~$a_i$, i.e.,~$g_i(1) = a_i$ or~$g_i(2) = a_i$.
With this construction every inner vertex has a twin and we keep only one copy.

Applying the gadget construction to suitable base graphs we obtain a family of new benchmark graphs that we call the \emph{dihedral construction}.

\paragraph{Unentwined gadgets.}
As a last example suppose~$\Gamma_1 = H_2 \times H_3$,~$\Gamma_2 = H_1\times H_3$ and~$\Gamma_3 = H_1\times H_2$ with arbitrary finite groups~$H_i$. 
Then the group~$\Delta \coloneqq \{((h_2,h_3),(h_1,h_3),(h_1,h_2))\}$ is a subgroup of~$S_n\times S_n\times S_n$. 
Figure~\ref{no:3:dependence:gadget} shows a gadget that realizes such a group for~$H_1\cong H_2\cong S_2$ and~$H_3 \cong S_3$.
We call these gadgets unentwined since they only create a pairwise interconnection between the classes. In particular, the gadgets are not beneficial for our cause since they do not force any sort of interplay between all three classes~$C_1$,~$C_2$ and~$C_3$.

\subsection{Group theoretic restrictions on the gadgets}

Suppose we want to realize a gadget~$X$ that has a certain automorphism behavior on the set of outer vertices~$C_1\cup C_2 \cup C_3$, where we assume the three classes~$C_i$ as invariant (i.e., they are color classes).
As discussed before, every group~$\Delta \leq \Sym(C_1)\times \Sym(C_2)\times \Sym(C_3)$ can be realized by such a gadget. 

Our intuition is that the isomorphism solvers can locally analyze a bounded size gadget and thus implicitly determine the automorphism group~$\Delta = \Aut(X)|_{C_1\cup C_2 \cup C_3}$ of such a gadget.
The goal must therefore be to encode the difficulty in the global dependency across the entire graph rather than in a local gadget. 

Note that~$\Delta \leq \pi_1(\Delta) \times \pi_2(\Delta) \times \pi_3(\Delta)$, where~$\pi_i$ is the projection on the~$i$-th component.
We say that~$\Delta$ is a \emph{subdirect product} of~$\pi_1(\Delta) \times \pi_2(\Delta) \times \pi_3(\Delta)$.
To understand what kind of gadgets can be constructed we should therefore investigate subdirect products.
We call such a group~$\Delta\leq \pi_1(\Delta) \times \pi_2(\Delta) \times \pi_3(\Delta)$ \emph{2-factor injective} if each projection onto two of the factors is injective and \emph{2-factor surjective} if each of these projections is surjective. 

Intuitively, 2-factor injectivity says that two components determine the third.
Thus when two different gadgets that are not 2-factor injective are attached (via parallel edges say) this may create local automorphism in the resulting graph, which we are trying to avoid. We are therefore interested in creating 2-factor injective gadgets.
In fact, it is possible to restrict ourselves to 2-factor injective gadgets altogether, since we can quotient out the elements responsible for non-injectivity.
(Formally we take the quotient of~$\Delta$ by the kernel of the projection onto two components.) 

In general, if~$\Gamma_1$ is an abelian group then the group~$\Delta := \{(a,b,c) \in (\Gamma_1)^{3}\mid abc = 1\}$, that was described as first example in the previous section, is~$2$-factor injective and~$2$-factor surjective.
It turns out that this is the only type of group with these properties. 

\begin{lemma}[\cite{group:theory:paper}]\label{inj:and:surj}
 Let $\Gamma = \Gamma_1\times \Gamma_2\times \Gamma_3$ be a finite group and~$\Delta$ a subdirect product of~$\Gamma$ that is~$2$-factor surjective and~$2$-factor injective.
 Then~$\Gamma_1$,~$\Gamma_2$ and~$\Gamma_3$ are isomorphic abelian groups and~$\Delta$ is isomorphic to the subgroup of~$\Gamma_1^3$ given by~$\{(a,b,c) \in (\Gamma_1)^{3}\mid abc = 1\}$, which in turn is isomorphic to~$(\Gamma_1)^{2}$ as an abstract group.
\end{lemma}

Thus, if we want to move beyond abelian groups, we cannot require 2-factor surjectivity.
In general, however, it can be shown using group theoretic means that every 2-factor injective group~$\Delta \leq \pi_1(\Delta) \times \pi_2(\Delta) \times \pi_3(\Delta)$ is obtained by a combination of the three example constructions described in the previous subsection.
Indeed, in \cite{group:theory:paper} it is in particular shown that if one wants an entwined gadget then it is necessary to use an abelian group of the form described in Lemma~\ref{inj:and:surj}.
One can then take a finite extension similar to the semidirect product construction described above.
We refer to~\cite{group:theory:paper} for more details.
Overall we conclude that there are some group theoretic obstructions to using non-abelian gadgets, and are left with the intuition that the original construction leads to the most difficult examples.

\section{Experimental Results}
\label{sec:benchmark}

In the following we detail experimental results for the constructions presented in this paper.
The experiments were each performed on single node of a compute cluster with 2.00 GHz Intel Xeon X5675 processors.
We always set a time limit of three hours (i.e., $10800$ seconds) and the memory limit to $4$ GB.
Every single instance was processed once, but multiple instances were generated for each possible number of vertices with the same construction.
We evaluated the following isomorphism solvers: \texttt{Bliss} version 0.72,  \texttt{Nauty/Traces} version 25r9, \texttt{Saucy} version 3.0, and \texttt{Conauto} version 2.03.

All our instances consist of two graphs and the task for the algorithms was to check for isomorphism.
For \texttt{Bliss} and \texttt{Nauty/Traces} this means that both graphs are canonized and then the canonical forms are compared for equality.
\texttt{Conauto} directly supports isomorphism testing and for \texttt{Saucy}, which only supports automorphism group computation, we compute the automorphism group of the disjoint union to check whether there is an isomorphism.
In each series we performed the graphs in increasing number of vertices with a time limit, which implied that once timeouts are reached repeatedly only a handful of further examples are computed.

We want to stress at this point that while the memory limit is irrelevant for \texttt{Bliss}, \texttt{Nauty}, \texttt{Saucy} and \texttt{Conauto}, it is significant for \texttt{Traces}.
In fact, in many larger instances \texttt{Traces} reaches the memory limit before reaching the time limit.
For the sake of readability we do not distinguish between those two cases and in our figures runs that reached the memory limit are displayed as reaching the time limit.  We refer to~\cite{mckay} for details on the memory usage of \texttt{Traces}.

The benchmark results are subdivided into three parts.
In the first part we consider the multipede graphs and their shrunken versions (i.e., the constructions from Sections~\ref{sec-rigid} and~\ref{sec:reduce-vertices}).
In the second part we consider the construction that employs gadgets for other groups (the construction from Section~\ref{sec:other:groups}).
This section concludes with a comparison to other available benchmark graphs.

We conclude with some remarks on our intuition for the benchmarks. While all constructions pose difficult challenges for the solvers, the ones that employ the more involved gadgets do not seem to yield more difficult examples.
A reason for this effect could lie in the size of the gadgets.
If the basic gadget is too large, asymptotic difficulty may emerge only for graphs with a number of vertices drastically larger than what could be tested.
In any case, all our constructions yield efficient methods to generate difficult isomorphism instances even with the size of the vertex set in regimes where no other difficult instances are available.

\subsection{Shrunken Multipedes}

Figure~\ref{fig:benchmark-rigid-base} shows the running times of the various isomorphism solvers on the multipede graphs described in Section~\ref{sec-rigid} without any reductions.
We observe that the multipede construction yields graphs which become infeasible for all solvers already for a few thousand vertices. Similarly, Figure~\ref{fig:benchmark-rigid-base:reduced} shows the running times on shrunken multipedes, i.e., the graphs to which the two node reductions of Section~\ref{sec:reduce-vertices} have been applied. In comparison we observe that similar running times to the unreduced graphs are obtained already on graphs that have roughly half the number of vertices. While the shrunken version results in more difficult graphs one can also see that there is a significantly larger fluctuation (presumably depending on~$\sigma$) among the graphs on a fixed number of vertices.

\begin{figure}[H]
 \centering
 \begin{tikzpicture}[scale =0.9]
  \begin{axis}[grid = both, xlabel = {number of vertices}, ylabel = {$10^4$ sec}, legend entries = {Bliss, Traces, Nauty, Saucy, Conauto},
               legend style = {at = {(0.05, 0.9)}, anchor = north west}, extra y ticks={10800}, extra y tick labels={timeout}, extra tick style={grid=major}, cycle list name = color list]
   \addplot+[only marks] table{benchmarks/z2-bliss.dat};
   \addplot+[only marks] table{benchmarks/z2-traces.dat};
   \addplot+[only marks] table{benchmarks/z2-nauty.dat};
   \addplot+[only marks] table{benchmarks/z2-saucy.dat};
   \addplot+[only marks] table{benchmarks/z2-conauto.dat};
  \end{axis}
 \end{tikzpicture}
 \begin{tikzpicture}[scale =0.9]
  \begin{axis}[ xmax=3200,  ymode = log, grid = both, xlabel = {number of vertices}, ylabel = {sec},
               legend style = {at = {(0.05, 0.9)}, anchor = north west}, 
               cycle list name = color list]
   \addplot+[only marks] table{benchmarks/z2-bliss.dat};
   \addplot+[only marks] table{benchmarks/z2-traces.dat};
   \addplot+[only marks] table{benchmarks/z2-nauty.dat};
   \addplot+[only marks] table{benchmarks/z2-saucy.dat};
   \addplot+[only marks] table{benchmarks/z2-conauto.dat};
  \end{axis}
 \end{tikzpicture}
 \caption[Experiments for multipedes]{Performance of various algorithms on $R(B(G_n,\sigma))$ for random permutations $\sigma$ in linear (left) and logarithmic scale (right).}
 \label{fig:benchmark-rigid-base}
\end{figure}
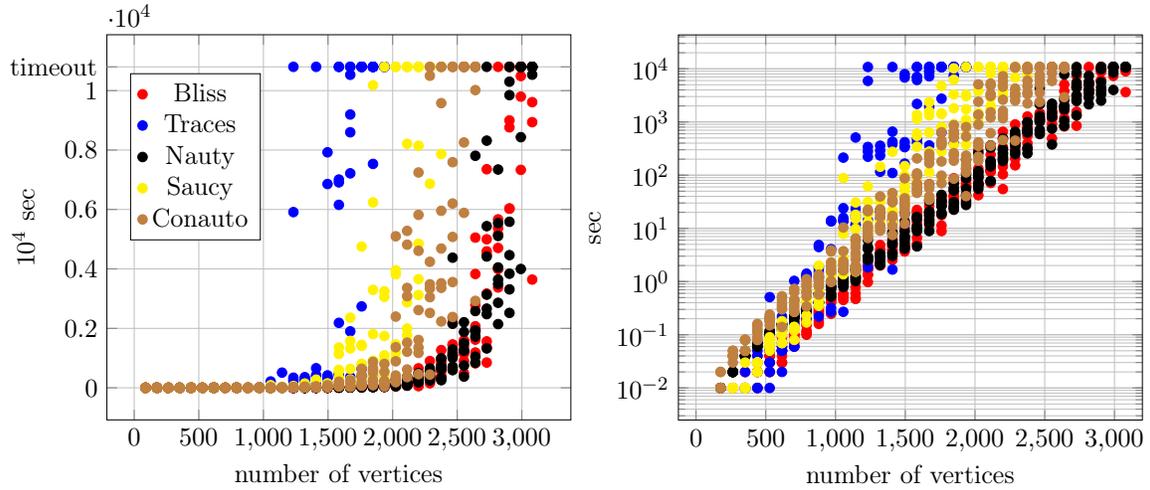

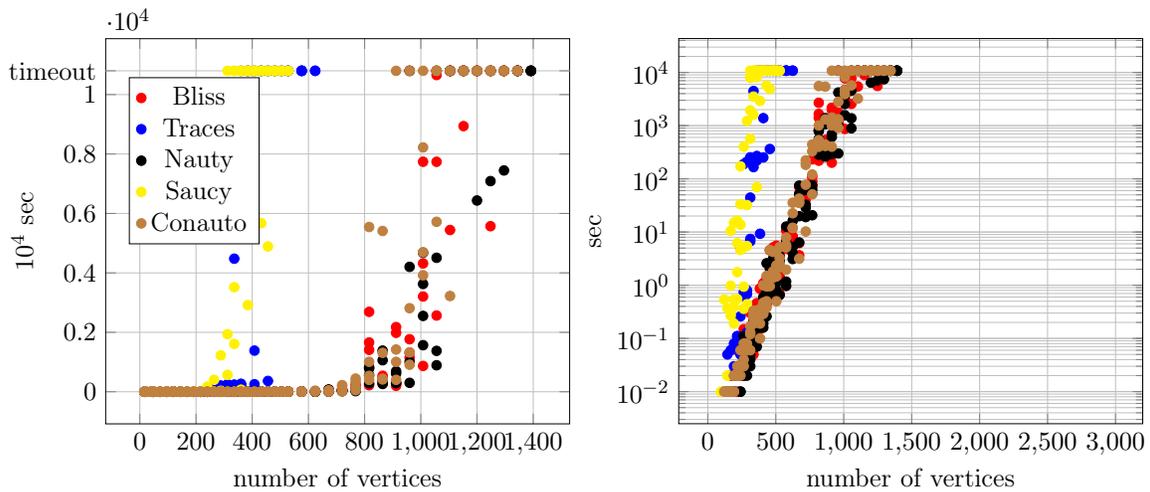
\begin{figure}[H]
 \centering
 \begin{tikzpicture}[scale =0.9]
  \begin{axis}[grid = both, xlabel = {number of vertices}, ylabel = {$10^4$ sec}, legend entries = {Bliss, Traces, Nauty, Saucy, Conauto},
               legend style = {at = {(0.05, 0.9)}, anchor = north west}, extra y ticks={10800}, extra y tick labels={timeout}, extra tick style={grid=major}, cycle list name = color list]
   \addplot+[only marks] table{benchmarks/t2-bliss.dat};
   \addplot+[only marks] table{benchmarks/t2-traces.dat};
   \addplot+[only marks] table{benchmarks/t2-nauty.dat};
   \addplot+[only marks] table{benchmarks/t2-saucy.dat};
   \addplot+[only marks] table{benchmarks/t2-conauto.dat};
  \end{axis}
 \end{tikzpicture}
 \begin{tikzpicture}[scale =0.9]
  \begin{axis}[xmax=3200, ymode = log, grid = both, xlabel = {number of vertices}, ylabel = {sec}, cycle list name = color list]
   \addplot+[only marks] table{benchmarks/t2-bliss.dat};
   \addplot+[only marks] table{benchmarks/t2-traces.dat};
   \addplot+[only marks] table{benchmarks/t2-nauty.dat};
   \addplot+[only marks] table{benchmarks/t2-saucy.dat};
   \addplot+[only marks] table{benchmarks/t2-conauto.dat};
  \end{axis}
 \end{tikzpicture}
 \caption[Experiments for shrunken multipedes]{Performance of various algorithms on the shrunken multipedes $R^*(B^*(G_n,\sigma))$ (both reduction techniques applied) for random permutations $\sigma$ in linear (left) and logarithmic scale (right).}
 \label{fig:benchmark-rigid-base:reduced}
\end{figure}

The next series of experiments highlights the effects of the two methods for reducing the number of vertices outlined in Section~\ref{sec:reduce-vertices}.
Figure~\ref{fig:benchmark-shrunken} shows the running times on the graphs in which only the reduction of outer vertices has been applied and a comparison to the unreduced graphs.
Similarly, Figure~\ref{fig:benchmark-basis-cfi} shows the running times of the graphs in which
the linear algebra based reduction has been applied and also a comparison to the unreduced graphs.
Finally, Figure~\ref{fig:benchmark-one-red-vs-two} shows a comparison of the effect of using one of the reductions compared to using both. Overall we conclude that both reductions yield significantly smaller instances without making the inputs easier and a combination of the two reductions yields the best results.

\pgfplotscreateplotcyclelist{mycolorliststar}{%
red,every mark/.append style={fill=red!40!black,red!50!black},mark=star\\%
blue,every mark/.append style={fill=blue!80!black,blue!50!black},mark=star\\%
black,every mark/.append style={fill=black!80!black,black!50!black},mark=star\\%
yellow,every mark/.append style={fill=yellow!80!black,yellow!80!black},mark=star\\%
brown,every mark/.append style={fill=brown!80!black,brown!80!black},mark=star\\%
}

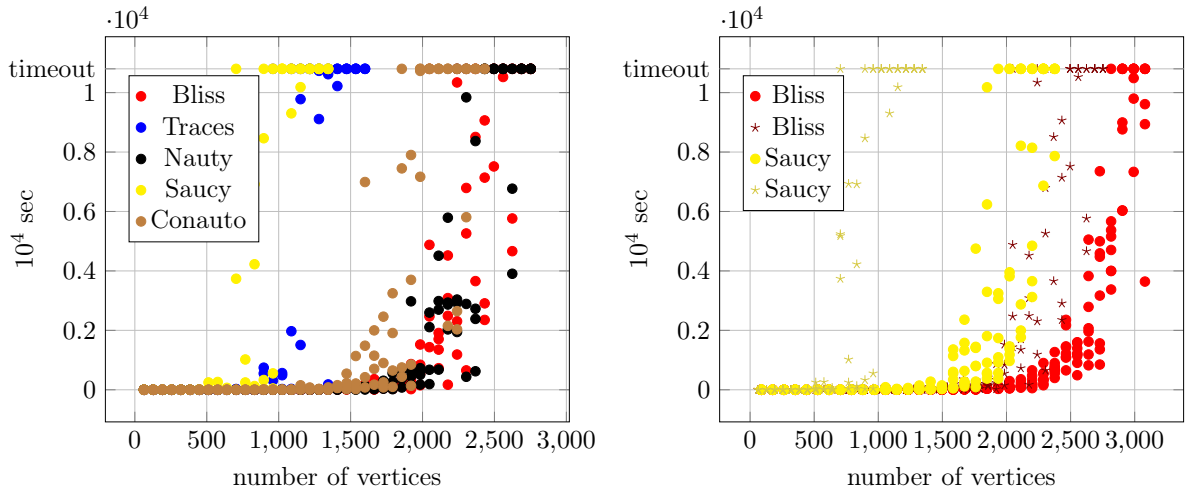
\begin{figure}[H]
 \centering
 \begin{tikzpicture}[scale = 0.9]
  \begin{axis}[grid = both, xlabel = {number of vertices}, ylabel = {$10^4$ sec}, legend entries = {Bliss, Traces, Nauty, Saucy, Conauto},
               legend style = {at = {(0.05, 0.9)}, anchor = north west}, extra y ticks={10800}, extra y tick labels={timeout}, extra tick style={grid=major}, cycle list name = color list]
   \addplot+[only marks] table{benchmarks/s2-bliss.dat};
   \addplot+[only marks] table{benchmarks/s2-traces.dat};
   \addplot+[only marks] table{benchmarks/s2-nauty.dat};
   \addplot+[only marks] table{benchmarks/s2-saucy.dat};
   \addplot+[only marks] table{benchmarks/s2-conauto.dat};
  \end{axis}
 \end{tikzpicture}
 \begin{tikzpicture}[scale = 0.9]
  \begin{axis}[grid = both, xlabel = {number of vertices}, ylabel = {$10^4$ sec}, legend entries = {Bliss, Bliss, Saucy, Saucy,a,b,c},
               legend style = {at = {(0.05, 0.9)}, anchor = north west}, extra y ticks={10800}, extra y tick labels={timeout}, extra tick style={grid=major}]
   \pgfplotsset{cycle list name = color list}
   \addplot+[only marks] table{benchmarks/z2-bliss.dat};
   \pgfplotsset{cycle list name=mycolorliststar}
   \pgfplotsset{cycle list shift=-1}
   \addplot+[only marks] table{benchmarks/s2-bliss.dat};
   \pgfplotsset{cycle list name=color list}
   \pgfplotsset{cycle list shift=+1}
   \addplot+[only marks] table{benchmarks/z2-saucy.dat};
   \pgfplotsset{cycle list name=mycolorliststar}
   \pgfplotsset{cycle list shift=0}
   \addplot+[only marks] table{benchmarks/s2-saucy.dat};
  \end{axis}
 \end{tikzpicture}
 \caption[Experiments for reduced base construction]{Performance of the various algorithms on the reduced graph $R^{*}(B(G_n,\sigma))$ (left) obtained by bypassing the outer vertices (see Subsection \ref{subsec:reduce-outer}).
          Also shown is a comparison for Bliss and Saucy on $R(B(G_n,\sigma))$ and $R^{*}(B(G_n,\sigma))$ (right, `$*$' indicates experiments for reduced graphs).}
 \label{fig:benchmark-shrunken}
\end{figure}

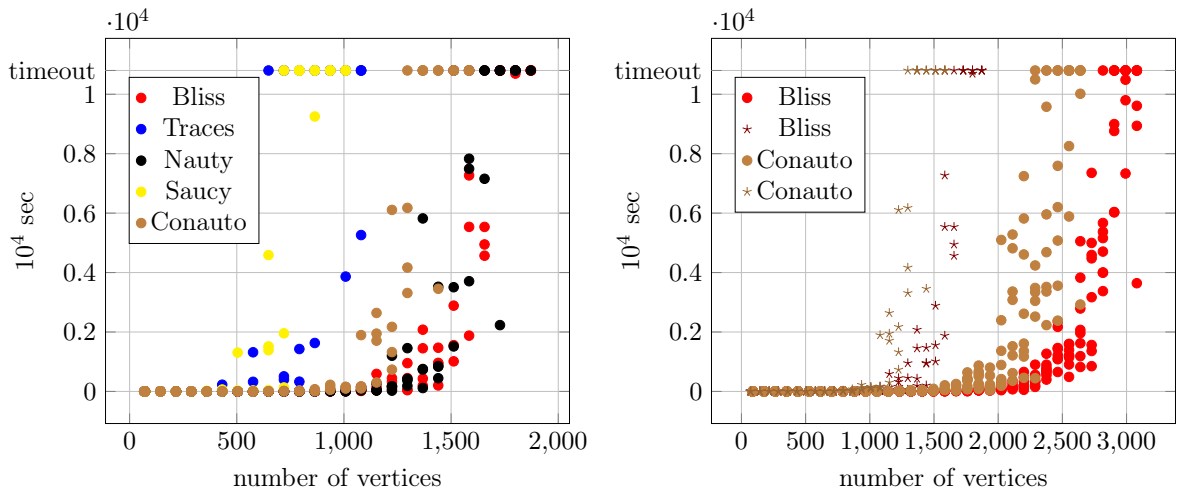
\begin{figure}[H]
 \centering
 \begin{tikzpicture}[scale =0.9]
  \begin{axis}[grid = both, xlabel = {number of vertices}, ylabel = {$10^4$ sec}, legend entries = {Bliss, Traces, Nauty, Saucy, Conauto},
               legend style = {at = {(0.05, 0.9)}, anchor = north west}, extra y ticks={10800}, extra y tick labels={timeout}, extra tick style={grid=major}, cycle list name = color list]
   \addplot+[only marks] table{benchmarks/r2-bliss.dat};
   \addplot+[only marks] table{benchmarks/r2-traces.dat};
   \addplot+[only marks] table{benchmarks/r2-nauty.dat};
   \addplot+[only marks] table{benchmarks/r2-saucy.dat};
   \addplot+[only marks] table{benchmarks/r2-conauto.dat};
  \end{axis}
 \end{tikzpicture}
 \begin{tikzpicture}[scale = 0.9]
  \begin{axis}[grid = both, xlabel = {number of vertices}, ylabel = {$10^4$ sec}, legend entries = {Bliss, Bliss, Conauto, Conauto},
               legend style = {at = {(0.05, 0.9)}, anchor = north west}, extra y ticks={10800}, extra y tick labels={timeout}, extra tick style={grid=major}]
   \pgfplotsset{cycle list name = color list}
   \addplot+[only marks] table{benchmarks/z2-bliss.dat};
   \pgfplotsset{cycle list name=mycolorliststar}
   \pgfplotsset{cycle list shift=-1}
   \addplot+[only marks] table{benchmarks/r2-bliss.dat};
   \pgfplotsset{cycle list name=color list}
   \pgfplotsset{cycle list shift=+2}
   \addplot+[only marks] table{benchmarks/z2-conauto.dat};
   \pgfplotsset{cycle list name=mycolorliststar}
   \pgfplotsset{cycle list shift=+1}
   \addplot+[only marks] table{benchmarks/r2-conauto.dat};
  \end{axis}
 \end{tikzpicture}
 \caption[Experiments for reduced base construction]{Performance of the various algorithms on the reduced graph $R(B^*(G_n,\sigma))$ (left) obtained by the linear algebra reduction (see Subsection \ref{subsec:reduce-linalg}).
          Also shown is a comparison for Bliss and Conauto on $R(B(G_n,\sigma))$ and $R(B^{*}(G_n,\sigma))$ (right, `$*$' indicates experiments for reduced graphs).}
 \label{fig:benchmark-basis-cfi}
\end{figure}

\begin{figure}[H]
 \centering
 \begin{tikzpicture}[scale = 0.9]
  \begin{axis}[grid = both, xlabel = {number of vertices}, ylabel = {$10^4$ sec}, legend entries = {Bliss, Bliss, Conauto, Conauto},
               legend style = {at = {(0.05, 0.9)}, anchor = north west}, extra y ticks={10800}, extra y tick labels={timeout}, extra tick style={grid=major}]
   \pgfplotsset{cycle list name = color list}
   \addplot+[only marks] table{benchmarks/r2-bliss.dat};
   \pgfplotsset{cycle list name=mycolorliststar}
   \pgfplotsset{cycle list shift=-1}
   \addplot+[only marks] table{benchmarks/t2-bliss.dat};
   \pgfplotsset{cycle list name=color list}
   \pgfplotsset{cycle list shift=+2}
   \addplot+[only marks] table{benchmarks/r2-conauto.dat};
   \pgfplotsset{cycle list name=mycolorliststar}
   \pgfplotsset{cycle list shift=+1}
   \addplot+[only marks] table{benchmarks/t2-conauto.dat};  
  \end{axis}
 \end{tikzpicture}
 \begin{tikzpicture}[scale = 0.9]
  \begin{axis}[grid = both, xlabel = {number of vertices}, ylabel = {$10^4$ sec}, legend entries = {Bliss, Bliss, Conauto, Conauto},
               legend style = {at = {(0.05, 0.9)}, anchor = north west}, extra y ticks={10800}, extra y tick labels={timeout}, extra tick style={grid=major}]
   \pgfplotsset{cycle list name = color list}
   \addplot+[only marks] table{benchmarks/s2-bliss.dat};
   \pgfplotsset{cycle list name=mycolorliststar}
   \pgfplotsset{cycle list shift=-1}
   \addplot+[only marks] table{benchmarks/t2-bliss.dat};
   \pgfplotsset{cycle list name=color list}
   \pgfplotsset{cycle list shift=+2}
   \addplot+[only marks] table{benchmarks/s2-conauto.dat};
   \pgfplotsset{cycle list name=mycolorliststar}
   \pgfplotsset{cycle list shift=+1}
   \addplot+[only marks] table{benchmarks/t2-conauto.dat};
  \end{axis}
 \end{tikzpicture}
 \caption[Experiments for reduced base construction]{Performance comparison for Bliss and Conauto for $R(B^*(G_n,\sigma))$ and $R^{*}(B^*(G_n,\sigma))$ (left) and for $R^*(B(G_n,\sigma))$ and $R^{*}(B^*(G_n,\sigma))$ (right) (`$*$' indicates experiments for reduced graphs).}
 \label{fig:benchmark-one-red-vs-two}
\end{figure}
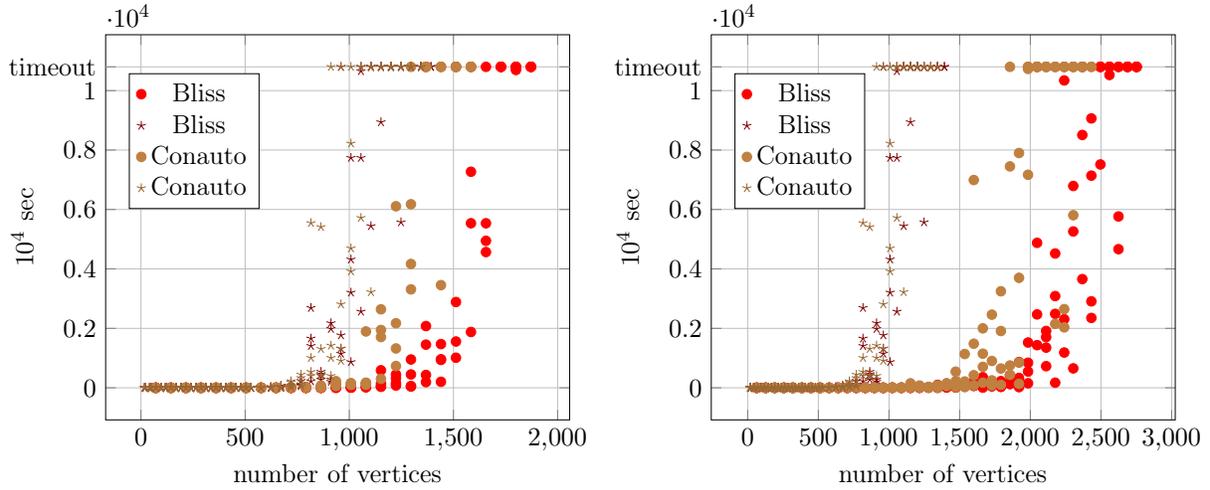

\subsection{Other groups}

In this second part of the section, we give some experiments for the rigid base construction on other groups (Figure~\ref{fig:benchmark-other-groups}).
We show here the performance of \texttt{Bliss} in order to compare the different constructions, but we remark that the behavior of the other programs is similar.
Overall, we observe in this part that the original construction based on $\mathbb{Z}_2$ yields the most difficult graphs.

\begin{figure}[H]
 \centering
 \begin{tikzpicture}[scale =0.9]
  \begin{axis}[grid = both, xlabel = {number of vertices}, ylabel = {$10^4$ sec}, legend entries = {$\mathbb{Z}_2$, $\mathbb{Z}_3$, $\mathbb{Z}_4$},
               legend style = {at = {(0.05, 0.9)}, anchor = north west}, at = {(0.5\linewidth,0)}, extra y ticks={10800},
               extra y tick labels={timeout}, extra tick style={grid=major},cycle list name = exotic]
   \addplot+[only marks] table{benchmarks/z2-bliss.dat};
   \addplot+[only marks] table{benchmarks/z3-bliss.dat};
   \addplot+[only marks] table{benchmarks/z4-bliss.dat};
  \end{axis}
 \end{tikzpicture}
 \begin{tikzpicture}[scale =0.9]
  \begin{axis}[grid = both, xlabel = {number of vertices}, ylabel = {$10^4$ sec}, legend entries = {$D_3$, $D_4$},
               legend style = {at = {(0.05, 0.9)}, anchor = north west}, extra y ticks={10800},
               extra y tick labels={timeout}, extra tick style={grid=major},cycle list name = exotic]
   \addplot+[only marks] table{benchmarks/d3-bliss.dat};
   \addplot+[only marks] table{benchmarks/d4-bliss.dat};
  \end{axis}
 \end{tikzpicture}
 \caption[Experiments with gadgets inducing other groups]{Performance of Bliss on the rigid base construction with various abelian groups $\mathbb{Z}_k$ (left) and the dihedral construction $D_k$ (right) as described in Section \ref{sec:other-groups-constructions}.}
 \label{fig:benchmark-other-groups}
\end{figure}
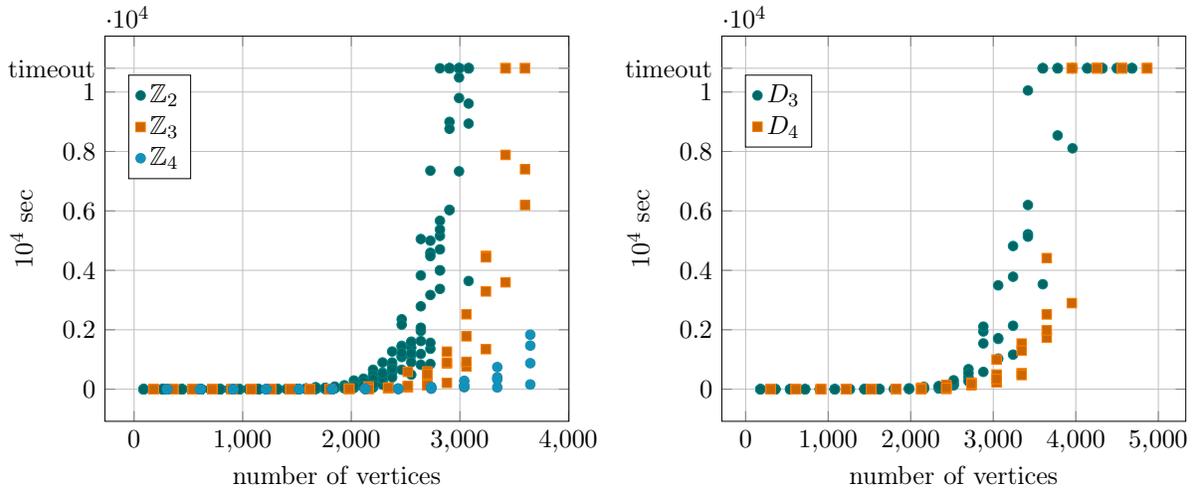

\subsection{Comparison with other benchmarks}

We compare the running times to benchmark graphs that previously existed.
Since the existing benchmark graphs typically do not come in non-isomorphic pairs of graphs we always take two copies, for each of which we randomly permute its vertices, and then perform an isomorphism test for the two copies.
This way, the results obtained for existing benchmark graphs are comparable to the results presented in the previous sections.
 
\paragraph{Comparison with other benchmark graphs.}  Concerning other benchmark graphs let us start by saying that the number of available graphs that one might call difficult in practice is quite limited.
Of the benchmark libraries mentioned in the introduction, the benchmark library available at \url{http://pallini.di.uniroma1.it/} is very comprehensive and explicitly includes all the other benchmark libraries we mentioned.
At that website an extensive suite of running times for all algorithms on the benchmark library is given.
All our running times on the benchmark graphs seem to be essentially proportional to the results presented at the web page.
Thus, the graphs that appear as most difficult on the web page are also the ones with the largest running time for us.
We provide comparisons for our graphs to the families of most difficult graphs from the library.
For a description of the graphs we refer to \url{http://pallini.di.uniroma1.it}.

\begin{itemize}

\item Combinatorial graphs of Gordan Royle. (\textsf{combinatorial})

The \textsf{combinatorial} family of combinatorial graphs collected by Gordan Royle is a collection of graphs that are challenging for the isomorphism solvers.
It consists of 12 graphs, some of which are quite large. We tested all algorithms on each graph, Table~\ref{tab:combin} summarizes the results.
The graphs are ordered by the number of vertices.
Some running times of the best algorithm on shrunken multipede graphs on a similar number of vertices are provided as comparisons up to 1500 vertices. (Beyond that there are only timeouts.)

\begin{table}
 \centering
 \scalebox{0.95}{
 \begin{tabular}{| l| r || r | r | r | r | r ||| r | r |}
  \hline
  \multicolumn{7}{|c|||}{combinatorial graphs} & \multicolumn{2}{c|}{$R^{*}(B^*(G_n,\sigma))$}\\
  \hline
  \multicolumn{1}{|c|}{graphs} & vertices & Bliss    & Traces & Nauty   & Saucy   & Conauto & vertices & Nauty   \\\hline
  X2                           & 756      & 656.46   & 0.04   & 484.62  & 290.00  & 0.30    & 720      & 19.1    \\\hline
  Pg25h1  \hfill (P1)          & 1302     & 10537.50 & 4.92   & -       & -       & 73.48   & 1296     & 7446.96 \\\hline
  flag6   \hfill (P2)          & 1514     & -        & 101.15 & -       & -       & 6419.26 & 1344     & -       \\\hline
  KM7LG                        & 2752     & -        & 1.41   & -       & -       & 16.81   &          &         \\\hline
  X1                           & 3276     & 144.49   & 0.33   & 2739.36 & -       & 1.21    &          &         \\\hline 
  X0                           & 3650     & 5640.84  & 14.77  & -       & -       & 8.64    &          &         \\\hline
  BG4                          & 3650     & 4221.71  & 6.07   & -       & 8965.94 & 4.40    &          &         \\\hline
  Fi9                          & 7300     & -        & 6.50   & -       & -       & 317.20  &          &         \\\hline
  Pg64     \hfill (P3)         & 8322     & -        & 902.26 & -       & -       & -       &          &         \\\hline
  FTWKB11                      & 15984    & -        & 35.24  & -       & -       & 1961.99 &          &         \\\hline
  PM11LG                       & 15984    & -        & 94.27  & -       & -       & -       &          &         \\\hline
  NCKL900K                     & 900000   & 10.09    & 5.91   & 13.82   & 15.03   & -       &          &         \\\hline
 \end{tabular}}
 \caption{Running times in seconds for the various algorithms on the \textsf{combinatorial} family and comparisons for shrunken multipedes}\label{tab:combin}
\end{table}

\item Projective Planes of order 25 and 27 (\textsf{pp})

The family \textsf{pp} contains Desarguesian and non-Desarguesian (i.e., algebraic and non-algebraic) projective planes of orders up to 27. We have only tested the larger planes of order 25 (1302 vertices) and order 27 (1514 vertices). Figure~\ref{fig:benchmark-pp} shows the results. (Note that three projective planes are contained in the \textsf{combinatorial} family). For comparison: the fastest running time on a shrunken multipede on 1296 vertices was 7446.96 seconds and on shrunken multipedes with 1344 vertices timeout occurs for all algorithms.

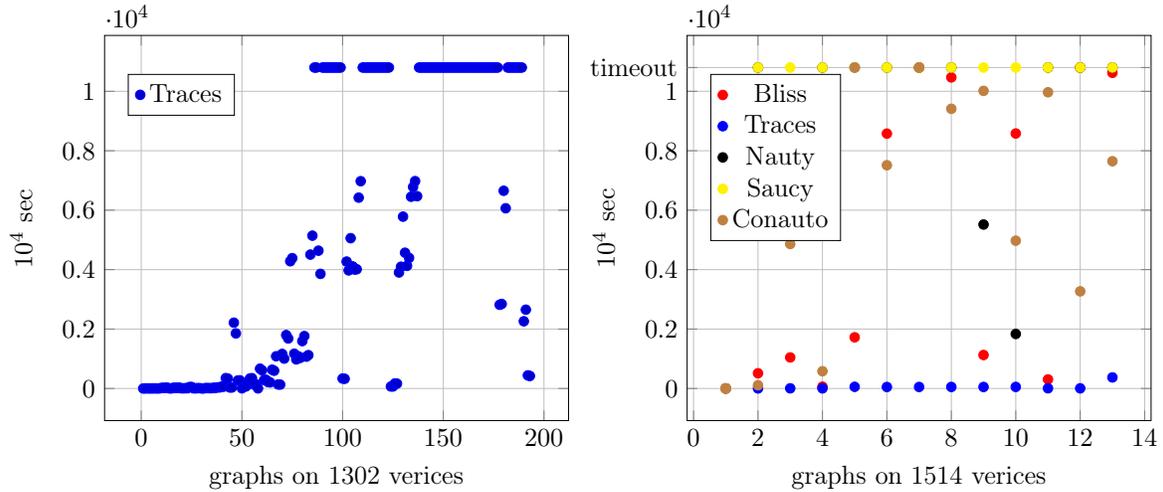
\begin{figure}[H]
 \centering
 \begin{tikzpicture}[scale = 0.9]
  \begin{axis}[grid = both, xlabel = {graphs on 1302 verices}, ylabel = {$10^4$ sec}, legend entries = {Traces},
               legend style = {at = {(0.05, 0.9)}, anchor = north west}, extra tick style={grid=major}]
   \addplot+[only marks] table{benchmarks/pp25-traces.dat};
  \end{axis}
 \end{tikzpicture}
 \begin{tikzpicture}[scale = 0.9]
  \begin{axis}[grid = both, xlabel = {graphs on 1514 verices}, ylabel = {$10^4$ sec}, legend entries = {Bliss, Traces, Nauty, Saucy, Conauto},
               legend style = {at = {(0.05, 0.9)}, anchor = north west}, extra y ticks={10800}, extra y tick labels={timeout}, extra tick style={grid=major}, cycle list name = color list]
   \addplot+[only marks] table{benchmarks/pp27-bliss.dat};
   \addplot+[only marks] table{benchmarks/pp27-traces.dat};
   \addplot+[only marks] table{benchmarks/pp27-nauty.dat};
   \addplot+[only marks] table{benchmarks/pp27-saucy.dat};
   \addplot+[only marks] table{benchmarks/pp27-conauto.dat};
  \end{axis}
 \end{tikzpicture}
 \caption{Running times on projective planes of order 25 (left) and order 27 (right).}
 \label{fig:benchmark-pp}
\end{figure}

\item Non-disjoint unions of tripartite graphs (\textsf{tnn})

The family \textsf{tnn} contains graphs which have so called color-components, which some of the algorithms are tuned to exploit. Figure~\ref{fig:benchmark-tnn} shows the running times and comparisons to running times on shrunken multipedes.

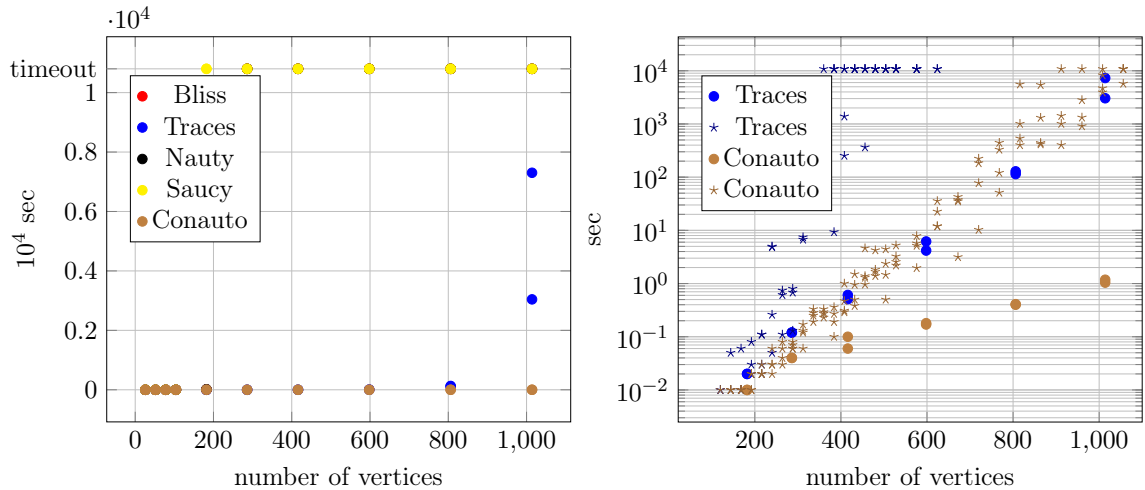
\begin{figure}[H]
 \centering
 \begin{tikzpicture}[scale = 0.9]
  \begin{axis}[grid = both, xlabel = {number of vertices}, ylabel = {$10^4$ sec}, legend entries = {Bliss, Traces, Nauty, Saucy, Conauto},
               legend style = {at = {(0.05, 0.9)}, anchor = north west}, extra y ticks={10800}, extra y tick labels={timeout}, extra tick style={grid=major}, cycle list name = color list]
   \addplot+[only marks] table{benchmarks/tnn-bliss.dat};
   \addplot+[only marks] table{benchmarks/tnn-traces.dat};
   \addplot+[only marks] table{benchmarks/tnn-nauty.dat};
   \addplot+[only marks] table{benchmarks/tnn-saucy.dat};
   \addplot+[only marks] table{benchmarks/tnn-conauto.dat};
  \end{axis}
 \end{tikzpicture}
 \begin{tikzpicture}[scale = 0.9]
  \begin{axis}[xmax=1100, grid = both,ymode = log, xlabel = {number of vertices}, ylabel = {sec}, legend entries = {Traces, Traces, Conauto, Conauto},
               legend style = {at = {(0.05, 0.9)}, anchor = north west}, extra tick style={grid=major}, cycle list name = color list]
   \pgfplotsset{cycle list shift=+1}
   \addplot+[only marks] table{benchmarks/tnn-traces.dat};
   \pgfplotsset{cycle list name=mycolorliststar}
   \pgfplotsset{cycle list shift=0}
   \addplot+[only marks] table{benchmarks/t2-traces.dat};
   \pgfplotsset{cycle list name = color list}
   \pgfplotsset{cycle list shift=+2}
   \addplot+[only marks] table{benchmarks/tnn-conauto.dat};
   \pgfplotsset{cycle list name=mycolorliststar}
   \pgfplotsset{cycle list shift=+1}
   \addplot+[only marks] table{benchmarks/t2-conauto.dat};
  \end{axis}
 \end{tikzpicture}
 \caption{Running times on unions of tripartite graphs (\textsf{tnn}) for the various algorithms (left) and a comparison with shrunken multipedes (right, `$*$' indicates experiments on the shrunken multipede construction).}
 \label{fig:benchmark-tnn}
\end{figure}

\item Cai-F\"{u}rer-Immerman graphs  (\textsf{cfi})

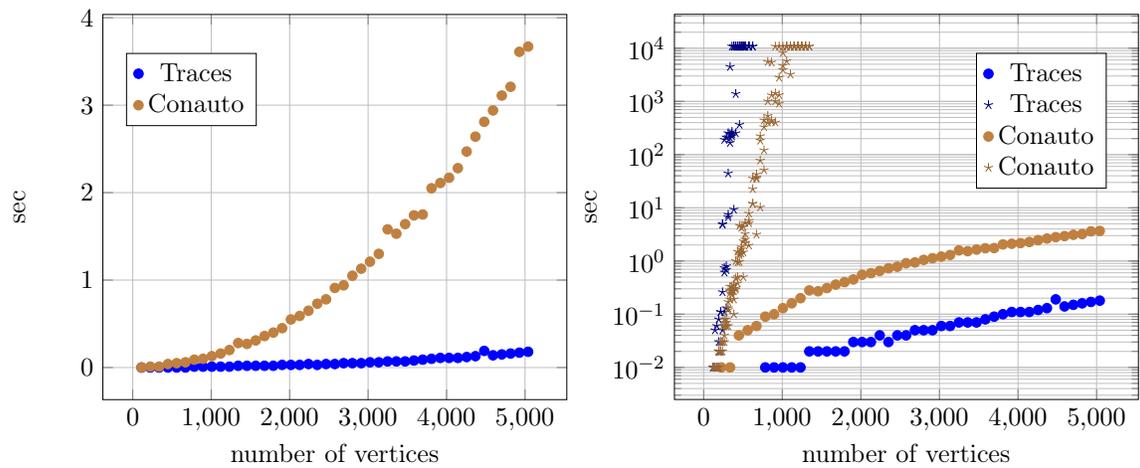
\begin{figure}[H]
 \centering
 \begin{tikzpicture}[scale = 0.9]
  \begin{axis}[grid = both, xlabel = {number of vertices}, ylabel = {sec}, legend entries = {Traces, Conauto},
               legend style = {at = {(0.05, 0.9)}, anchor = north west}, cycle list name = color list]
    \pgfplotsset{cycle list shift=+1}            
   \addplot+[only marks] table{benchmarks/c2-traces.dat};
    \pgfplotsset{cycle list shift=+3}
   \addplot+[only marks] table{benchmarks/c2-conauto.dat};
  \end{axis}
 \end{tikzpicture}
 \begin{tikzpicture}[scale = 0.9]
  \begin{axis}[grid = both,ymode = log, xlabel = {number of vertices}, ylabel = {sec}, legend entries = {Traces, Traces, Conauto, Conauto},
               legend style = {at = {(0.65, 0.9)}, anchor = north west}, extra tick style={grid=major}, cycle list name = color list]
   \pgfplotsset{cycle list shift=+1}
   \addplot+[only marks] table{benchmarks/c2-traces.dat};
   \pgfplotsset{cycle list name=mycolorliststar}
   \pgfplotsset{cycle list shift=0}
   \addplot+[only marks] table{benchmarks/t2-traces.dat};
   \pgfplotsset{cycle list name = color list}
   \pgfplotsset{cycle list shift=+2}
   \addplot+[only marks] table{benchmarks/c2-conauto.dat};
   \pgfplotsset{cycle list name=mycolorliststar}
   \pgfplotsset{cycle list shift=+1}
   \addplot+[only marks] table{benchmarks/t2-conauto.dat};
  \end{axis}
 \end{tikzpicture}
 \caption[Comparison for rigid base construction and the CFI-construction]{Performance of some algorithms on the CFI-construction (left). Note that the scale is in seconds not~$10^4$ seconds as for most other plots.
          Also shown is a comparison for Traces and Conauto between the CFI-construction and the rigid base construction (right, `$*$' indicates experiments on the shrunken multipede construction).}
 \label{fig:benchmark-standard-cfi}
\end{figure}

\item Miyazaki graphs (\textsf{mz-aug2})

\begin{figure}[H]
 \centering
 \begin{tikzpicture}[scale = 0.9]
  \begin{axis}[grid = both, xlabel = {number of vertices}, ylabel = {sec}, legend entries = {Bliss, Traces},
               legend style = {at = {(0.05, 0.9)}, anchor = north west}, cycle list name = color list]
   \addplot+[only marks] table{benchmarks/mz-aug2-bliss.dat};
   \addplot+[only marks] table{benchmarks/mz-aug2-traces.dat};
  \end{axis}
 \end{tikzpicture}
 \begin{tikzpicture}[scale = 0.9]
  \begin{axis}[grid = both,ymode = log, xlabel = {number of vertices}, ylabel = {sec}, legend entries = {Bliss, Bliss, Traces, Traces},
               legend style = {at = {(0.05, 0.9)}, anchor = north west}, extra tick style={grid=major}, cycle list name = color list]
   \pgfplotsset{cycle list shift=+0}
   \addplot+[only marks] table{benchmarks/mz-aug2-bliss.dat};
   \pgfplotsset{cycle list name=mycolorliststar}
   \pgfplotsset{cycle list shift=-1}
   \addplot+[only marks] table{benchmarks/t2-bliss.dat};
   \pgfplotsset{cycle list name = color list}
   \pgfplotsset{cycle list shift=-1}
   \addplot+[only marks] table{benchmarks/mz-aug2-traces.dat};
   \pgfplotsset{cycle list name=mycolorliststar}
   \pgfplotsset{cycle list shift=-2}
   \addplot+[only marks] table{benchmarks/t2-traces.dat};
  \end{axis}
 \end{tikzpicture}
 \caption{Running times of the augmented Miyazaki graphs (\textsf{mz-aug2}) for bliss and traces and a comparison with times for shrunken multipedes (right, `$*$' indicates experiments on the shrunken multipedes construction).}
 \label{fig:benchmark-miyazaki}
\end{figure}
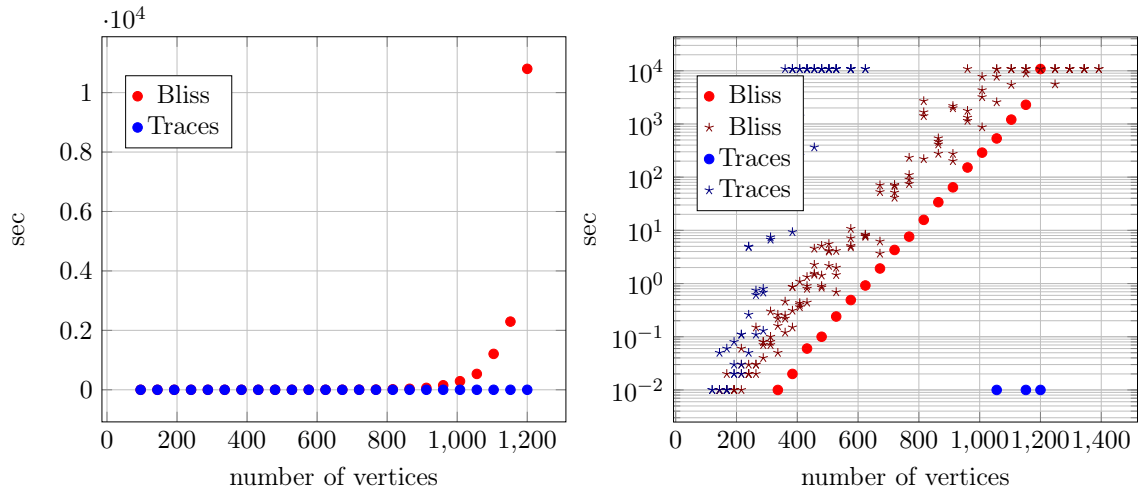

\end{itemize}

\section{Discussion}

We conclude with a discussion of the experimental results. Let us first address the question of variance in the experiments. 
 
\paragraph{Variance.} 
For each graph we tested each algorithm only once. Not all of the algorithms are deterministic and even for deterministic algorithms, there is also the question whether permuting the input graphs has any influence on the running times.
Table~\ref{tab:variance} shows experiments in which we ran permutations of the same graph several times for each of the algorithms.
Only saucy exhibits a non-negligible variance in some of the runs.
However, even that variance is negligible in comparison to the exponential scaling of the running times on the benchmark graphs as the number of vertices increases. 
 
\begin{table}
 \centering
 \scalebox{0.85}{ 
 \begin{tabular}{| c || r | r | r | r | r |}
  \hline
  \multicolumn{6}{|c|}{$R(B^*(G_n,\sigma))$}\\
  \hline
  $n$  & Bliss  & Traces & Nauty  & Saucy & Conauto\\
  \hline
       & 0.01   & 0.07   & 0.02   & 0.79  & 0.07\\
       & 0.01   & 0.07   & 0.02   & 3.12  & 0.07\\
  432  & 0.01   & 0.07   & 0.02   & 34.37 & 0.07\\
       & 0.01   & 0.07   & 0.02   & 0.41  & 0.07\\
       & 0.01   & 0.07   & 0.03   & 0.39  & 0.07\\
  \hline     
       & 1.03   & 915.63 & 3.25   & -     & 64.85\\
       & 1.03   & 892.14 & 3.23   & -     & 65.59\\
  864  & 1.02   & 953.45 & 3.28   & -     & 57.22\\
       & 1.03   & 934.35 & 3.28   & -     & 54.40\\
       & 1.03   & 882.44 & 3.28   & -     & 70.34\\
  \hline
       & 116.72 & -      & 20.15  & -     & 983.81\\
       & 124.24 & -      & 20.35  & -     & 948.19\\
  1224 & 116.51 & -      & 20.23  & -     & 901.45\\
       & 115.26 & -      & 20.29  & -     & 1092.45\\
       & 122.42 & -      & 20.01  & -     & 1423.21\\
  \hline
       & 3523.82& -      & 5527.81& -     & -\\
       & 2493.47& -      & 5532.19& -     & -\\
  1656 & 2670.06& -      & 5517.05& -     & -\\
       & 3358.18& -      & 5534.84& -     & -\\
       & 4156.87& -      & 5560.01& -     & -\\
  \hline
 \end{tabular}}
 % \scalebox{0.85}{
 % \begin{tabular}{| c || r | r | r | r | r |}
 %  \hline
 %  \multicolumn{6}{|c|}{$R^{*}(B(G_n,\sigma))$}\\
 %  \hline
 %  $n$  & Bliss  & Traces & Nauty  & Saucy  & Conauto\\
 %  \hline
 %       & 0.08   & 0.16   & 0.06   & 60.48  & 0.10\\
 %       & 0.08   & 0.16   & 0.04   & 17.68  & 0.10\\
 %  512  & 0.08   & 0.16   & 0.04   & 10.68  & 0.12\\
 %       & 0.08   & 0.18   & 0.04   & 261.67 & 0.11\\
 %       & 0.08   & 0.16   & 0.04   & 1.22   & 0.10\\
 %  \hline     
 %       & 1.82   & 326.78 & 1.06   & -      & 2.42\\
 %       & 1.87   & 352.33 & 0.67   & -      & 2.43\\
 %  1024 & 1.80   & 329.29 & 0.87   & -      & 2.39\\
 %       & 1.91   & 321.19 & 0.86   & 2502.53& 2.44\\
 %       & 1.78   & 336.98 & 1.06   & -      & 2.60\\
 %  \hline
 %       & 43.07  & -      & 29.53  & -      & 211.75\\
 %       & 42.12  & -      & 29.55  & -      & 212.30\\
 %  1536 & 42.14  & -      & 29.50  & -      & 213.02\\
 %       & 44.26  & -      & 29.72  & -      & 212.20\\
 %       & 43.97  & -      & 29.57  & -      & 213.10\\
 %  \hline
 %       & 1031.03& -      & 669.34 & -      & -\\
 %       & 1043.23& -      & 675.65 & -      & -\\
 %  2048 & 1058.69& -      & 668.44 & -      & -\\
 %       & 928.08 & -      & 700.51 & -      & -\\
 %       & 874.93 & -      & 672.11 & -      & -\\
 %  \hline
 % \end{tabular}}
 \scalebox{0.85}{
 \begin{tabular}{| c || r | r | r | r | r |}
  \hline
  \multicolumn{6}{|c|}{shrunken multipedes}\\
  \hline
  $n$  & Bliss  & Traces & Nauty  & Saucy  & Conauto\\
  \hline
       & 0.04   & 174.83 & 0.02   & -      & 0.05\\
       & 0.03   & 180.07 & 0.01   & -      & 0.05\\
  288  & 0.03   & 174.79 & 0.01   & 0.22   & 0.04\\
       & 0.03   & 174.41 & 0.01   & 0.22   & 0.04\\
       & 0.03   & 174.14 & 0.01   & 0.22   & 0.04\\
  \hline     
       & 13.32  & -      & 1.04   & -      & 5.08\\
       & 13.15  & -      & 0.61   & -      & 5.13\\
  576  & 13.43  & -      & 0.43   & -      & 5.82\\
       & 11.90  & -      & 0.57   & -      & 5.12\\
       & 10.13  & -      & 0.53   & -      & 5.05\\
  \hline
       & 508.76 & -      & 276.88 & -      & 462.12\\
       & 358.57 & -      & 278.75 & -      & 403.85\\
  864  & 342.34 & -      & 278.43 & -      & 424.83\\
       & 487.37 & -      & 279.75 & -      & 524.89\\
       & 403.34 & -      & 279.46 & -      & 408.64\\
  \hline
       & -      & -      & -      & -      & -\\
       & -      & -      & -      & -      & -\\
  1152 & -      & -      & -      & -      & -\\
       & -      & -      & -      & -      & -\\
       & -      & -      & -      & -      & -\\
  \hline
 \end{tabular}}
 \caption{Repeated runs on the same input graphs show that the running time of algorithms has a low variance on the benchmark graphs.}\label{tab:variance}
\end{table}
 
\paragraph{Input size.}
The experiments we present all order the graphs according to the number of vertices. However, it is a reasonable question whether one should rather consider the number of edges of the input graphs as the input size.
Indeed, in practice one may be interested in dealing with sparse graphs, and the various algorithms are in particular tuned for this case.
While all our graphs are sparse in the sense that they only have a linear number of edges, 
the average degree does vary. For shrunken multipedes we provided the coarse bound of~24 for the average degree. But the multipedes have an average degrees of 4.363, while the graphs obtained by only applying the linear algebra reduction ($R(B^{*}(G_n,\sigma))$) have an average degree of 4. It seems thus that when input size is measured in terms of the number of edges the graphs $R(B^{*}(G_n,\sigma))$ are the most difficult.
  
\paragraph{Conclusion.}
Overall we conclude that the new construction constitutes a simple algorithm that yields the most difficult benchmark graphs to date and experimentally we observe exponential behaviour in terms of the running times of practical isomorphism solvers.

\section*{Acknowledgments}
We thank Luis N\'{u}\~nez Chiroque, Tommi Junttila, Petteri Kaski,  Brendan McKay, Adolfo Piperno, and Jos\'{e} Luis L\'{o}pez-Presa for numerous discussions related to graph isomorphism and for feedback on our benchmark instances. This work was funded by the excellence initiative of the German federal and state governments.

\bibliographystyle{abbrv}
\bibliography{literature}

\newpage
\begin{appendix}
\noindent{\huge{\textbf{Appendix}}}
\section{Proof of Theorem~\ref{thm:construction:is:rigid}}\label{app:proof:of:rigidity:theorem}

This section is devoted to proving that with high probability the graphs obtained with the multipede construction are rigid.

\begin{reptheorem}{thm:construction:is:rigid}
The probability that $R(B(G_n, \sigma))$ is not rigid is in $\mathcal{O}\left(\frac{\log^{2}n}{n}\right)$.
\end{reptheorem}

\subsection{Basic inequalities}
\label{app:inequal}

To bound probabilities we use various basic inequalities that we summarize next. 
As usual, for $n,k \in \mathbb{N}$, $k \leq n$, the \emph{binomial coefficient} $\binom{n}{k} = \frac{n!}{k! \cdot (n-k)!}$ is the number of $k$-element subsets of an $n$-element set.

\begin{lemma}
 For $n_1,n_2,k_1,k_2 \in \mathbb{N}$ with $k_1 \leq n_1$ and~$k_2\leq n_2$, it holds that
 \begin{equation}
  \label{eq:binom-mult}
  \binom{n_1}{k_1}\cdot\binom{n_2}{k_2} \leq \binom{n_1+n_2}{k_1+k_2}.
 \end{equation}
\end{lemma}

\begin{lemma}
 \label{la:two-binoms-bound}
 For $n_1,n_2,k \in \mathbb{N}$ with $k \leq n_1 \leq n_2$ it holds that 
 \begin{equation}
  \label{eq:binom-inverse}
  \binom{n_1}{k} \cdot \binom{n_2}{k}^{-1} \leq \left(\frac{n_1}{n_2}\right)^{k}.
 \end{equation}
\end{lemma}

\begin{proof}
 We have
 \[\binom{n_1}{k} \cdot \binom{n_2}{k}^{-1} = \frac{\prod_{i=0}^{k-1}(n_1 - i)}{\prod_{i=0}^{k-1}(n_2 - i)} = \prod_{i=0}^{k-1}\left(\frac{n_1-i}{n_2 -i}\right).\]
 Furthermore $(n_1 - i)\cdot(n_2-i)^{-1} \leq n_1n_2^{-1}$ for all $i \leq n_1$ which gives the desired bound.
\end{proof}

\begin{lemma}
 For $n \in \mathbb{N}$ it holds that
 \begin{equation}
  \label{eq:binom-3n-n}
  \binom{3n}{n} \geq \frac{6^{n}}{2\sqrt{n}}.
 \end{equation}
\end{lemma}

\begin{proof}
 Using Sterling's approximation we get that \[\binom{2n}{n} \geq \frac{2^{2n-1}}{\sqrt{n}}.\]
 Combining this with Lemma \ref{la:two-binoms-bound} we obtain the desired bound.
\end{proof}

\begin{lemma}
 For $q < 1$ and $n \in \mathbb{N}$ it holds that \[\sum_{k=0}^{n}q^{k}k^{2} = \frac{q(q+1)}{(1-q)^{3}} - \frac{q^{n+1}}{(1-q)^{3}}\left(n^{2}(1-q)^{2} + 2n(1-q) + q + 1\right).\]
\end{lemma}

\begin{proof}
 We have \[\sum_{k=0}^{n}q^{k}k^{2} = \sum_{k=0}^{n} q\cdot\frac{d}{dq}\left(q\cdot\frac{d}{dq}q^{k}\right) = q\cdot\frac{d}{dq}\left(q\cdot\frac{d}{dq}\sum_{k=0}^{n} q^{k}\right) = q\cdot\frac{d}{dq}\left(q\cdot\frac{d}{dq}\frac{1 - q^{n+1}}{1-q}\right), \] and differentiating the right side proves the equation.
\end{proof}

\begin{corollary}
 For $q < 1$ and $n,\ell \in \mathbb{N}$ with $\ell < n$ it holds that
 \begin{equation}
  \label{eq:sum}
  \sum_{k=\ell+1}^{n}q^{k}k^{2} \leq \frac{q^{\ell+1}}{(1-q)^{3}}\left(\ell^{2}(1-q)^{2} + 2\ell(1-q) + q + 1\right).
 \end{equation}
\end{corollary}

\begin{lemma}
 For $n \in \mathbb{N}$, $n \geq 1$, it holds that
 \begin{equation}
  \label{eq:sqrt}
  \sqrt[n]{n} = n^{1/n} < 2.
 \end{equation}
\end{lemma}

\begin{proof}
 Let $f(x) = \sqrt[x]{x} = e^{\frac{\ln x}{x}}$ where $\ln$ is the natural logarithm with base $e$.
 Then the derivative is \[f'(x) = e^{\frac{\ln x}{x}} \frac{1- \ln x}{x^{2}} = x^{\frac{1}{x} - 2}(1 - \ln x).\]
 Hence $f'(x) < 0$ and $f$ is monotonically decreasing for all $x > e$.
 Furthermore $f(2) = \sqrt{2} < 2$ and $f(3) = \sqrt[3]{3} < 1.45 < 2$.
\end{proof}

\subsection{The probability for rigidity}

We call a subset of the edges~$C\subseteq E$ of some graph~$G$ \emph{2-regular (in~$G$)} if the graph induced by the edges is 2-regular. (I.e., every vertex of~$G$ is incident with exactly~$0$ or~$2$ edges of~$C$.)

\begin{lemma}
 \label{la-2-reg}
 $B(G,\sigma)$ is even if and only if there is a 2-regular set of edges~$C \subseteq E(G)$ in~$G$ such that $\sigma(C)$ is also 2-regular in $G$.
\end{lemma}

\begin{proof}
 Let $C \subseteq E(G)$ be a 2-regular subset of~$G$ such that $\sigma(C)$ is also 2-regular in~$G$. We set $X = C$.
 Then $|X \cap N(v)|$ is even for every $v \in V(G) \times \{0\}$, since $C$ is 2-regular, and even for every $v \in V(G) \times \{1\}$, since $\sigma(C)$ is 2-regular.
 
 Conversely, suppose that~$X\subseteq W_B$ is a witness that~$G$ is even. Then for every~$v\in  V(G) \times \{0\}$ by definition~$|X\cap N(v)|$ is even. This implies that an even number of edges in~$X$ are incident with~$v$. Since~$G$ is three regular we conclude that either 2 or 0 edges of~$X$ are incident with~$v$ and thus~$X$ is 2-regular in~$G$. Analogously we conclude that~$\sigma(X)$ is 2-regular in~$G$.
\end{proof}

\begin{lemma}
 \label{la:prob-odd}
 The probability that $B(G_n, \sigma)$ is not odd is in $\mathcal{O}\left(\frac{log^{2}n}{n}\right)$.
\end{lemma}

\begin{proof}
 Let $\sigma$ be a random permutation of $E(G_n)$.
 Using Lemma \ref{la-2-reg} it suffices to estimate the probability that there is some 2-regular edge set $C \subseteq E(G)$ such that $\sigma(C)$ is also 2-regular in~$G_n$.
 Note that every 2-regular edge set has at least~$4$ and at most~$2n$ edges and it also has at most half as many diagonals as other edges.
 For~$0\leq k \leq n$ and~$4\leq \ell \leq 2n$ let $n(k,\ell)$ be the number of 2-regular edge sets~$C \subseteq E(G)$ such that $|C| = \ell$ and $C$ contains exactly $k$ diagonal edges of $G_n$.
 Note that if we fix $k$ diagonal edges of $G_n$ then there are at most two 2-regular subgraphs containing exactly these diagonal edges.
 So $n(k,\ell) \leq 2 \cdot \binom{n}{k}$.
 For $\ell \leq 2\log n$ we can improve this bound to $n(k,\ell) \leq 2n^{k/2}(4\log n)^{k/2} = 2 \cdot (4n\log n)^{k/2}$.
 This can be seen from the fact that the diagonal edges can be paired so that each pair is of distance at most $2\log n$ (in particular~$k$ must be even).
 
 The probability that for particular two sets~$C$ and~$C'$ with~$|C| = |C'| = \ell$ a random permutation~$\sigma$ satisfies~$\sigma(C) = C'$ is~$\binom{3n}{\ell}^{-1}$.
 Thus, we can bound the probability $p$ that there is some 2-regular $C \subseteq E(G)$ such that $\sigma(C)$ is also 2-regular by
 \[p \leq \sum_{\ell = 4}^{2n}\sum_{k=0}^{\lfloor\frac{\ell}{2}\rfloor} \sum_{k'=0}^{\lfloor\frac{\ell}{2}\rfloor} n(k,\ell) \cdot n(k',\ell) \cdot \binom{3n}{\ell}^{-1}.\]
 To analyze the right side of the equation we divide  the first sum into three parts ranging from~$4$ to~$2\lfloor \log n \rfloor$ from~$2\lfloor \log n \rfloor+1$ to~$n$ and from~$n+1$ to~$2n$.
 For the first part we have
 \begin{align*}
       &\;\sum_{\ell = 4}^{2\lfloor \log n \rfloor}\sum_{k=0}^{\lfloor\frac{\ell}{2}\rfloor} \sum_{k'=0}^{\lfloor\frac{\ell}{2}\rfloor} n(k,\ell) \cdot n(k',\ell) \cdot \binom{3n}{\ell}^{-1}\\
  \leq &\;\sum_{\ell = 4}^{2\lfloor \log n \rfloor}\sum_{k=0}^{\lfloor\frac{\ell}{2}\rfloor} \sum_{k'=0}^{\lfloor\frac{\ell}{2}\rfloor} 4\cdot (4n\log n)^{(k+k')/2} \cdot \binom{3n}{\ell}^{-1}\\
  \leq &\;\sum_{\ell = 4}^{2\lfloor \log n \rfloor}4\left(\frac{\ell}{2} + 1\right)^{2}\cdot (4n \log n)^{\ell/2} \cdot \binom{3n}{\ell}^{-1}\\
  \leq &\;\sum_{\ell = 4}^{2\lfloor \log n \rfloor} 4\left(\frac{\ell}{2} + 1\right)^{2}\cdot (4n\log n)^{\ell/2}\cdot\frac{\ell^\ell}{(3n)^{\ell}}\\
  \leq &\;\sum_{\ell = 4}^{2\lfloor \log n \rfloor}\frac{4 \cdot \ell^{\ell+2}\cdot (2\log n)^{\ell/2}}{(3n)^{\ell/2}}\\
  \leq &\;\sum_{\ell = 4}^{2\lfloor \log n \rfloor}\frac{4 \cdot (2\log n)^{2\ell}}{(3n)^{\ell/2}}\\
  \leq &\;\frac{4 \cdot (2\log n)^{8}}{(3n)^{2}} \cdot \sum_{\ell = 0}^{2\lfloor \log n \rfloor - 4}\frac{(2\log n)^{2\ell}}{(3n)^{\ell/2}}\\
  =    &\;\mathcal{O}\left(\frac{\log^{2}n}{n}\right).
 \end{align*}
 For the second part we get
 \begin{align*}
       &\;\sum_{\ell = 2\lfloor \log n \rfloor + 1}^{n}\sum_{k=0}^{\lfloor\frac{\ell}{2}\rfloor} \sum_{k'=0}^{\lfloor\frac{\ell}{2}\rfloor} n(k,\ell) \cdot n(k',\ell) \cdot \binom{3n}{\ell}^{-1}\\
  \leq &\;\sum_{\ell = 2\lfloor \log n \rfloor + 1}^{n}\sum_{k=0}^{\lfloor\frac{\ell}{2}\rfloor} \sum_{k'=0}^{\lfloor\frac{\ell}{2}\rfloor} 2\binom{n}{k} \cdot 2\binom{n}{k'} \cdot \binom{3n}{\ell}^{-1}\\
  \stackrel{(\ref{eq:binom-mult})}{\leq} &\;\sum_{\ell = 2\lfloor \log n \rfloor + 1}^{n}\sum_{k=0}^{\lfloor\frac{\ell}{2}\rfloor} \sum_{k'=0}^{\lfloor\frac{\ell}{2}\rfloor} 4\binom{2n}{k+k'} \cdot \binom{3n}{\ell}^{-1}\\
  \leq &\;\sum_{\ell = 2\lfloor \log n \rfloor + 1}^{n}\sum_{k=0}^{\lfloor\frac{\ell}{2}\rfloor} \sum_{k'=0}^{\lfloor\frac{\ell}{2}\rfloor} 4\binom{2n}{\ell} \cdot \binom{3n}{\ell}^{-1}\\
  \leq &\;\sum_{\ell = 2\lfloor \log n \rfloor + 1}^{n} 4\left(\frac{\ell}{2} + 1\right)^{2} \cdot \binom{2n}{\ell} \cdot \binom{3n}{\ell}^{-1}\\
  \stackrel{(\ref{eq:binom-inverse})}{\leq} &\;\sum_{\ell = 2\lfloor \log n \rfloor + 1}^{n} 4\left(\frac{\ell}{2} + 1\right)^{2} \cdot \left(\frac{2}{3}\right)^{\ell}\\
  \leq &\;\sum_{\ell = 2\lfloor \log n \rfloor + 1}^{n} 4\ell^{2} \cdot \left(\frac{2}{3}\right)^{\ell}\\
  \stackrel{(\ref{eq:sum})}{\leq} &\;\left(\frac{2}{3}\right)^{2\lfloor \log  n \rfloor + 1}\left(48\lfloor \log n \rfloor^{2} + 144\lfloor \log n \rfloor + 180\right)\\
  \leq &\;\frac{1}{n}\left(48\lfloor \log n \rfloor^{2} + 144\lfloor \log n \rfloor + 180\right)\\
  =    &\;\mathcal{O}\left(\frac{\log^{2}n}{n}\right).
 \end{align*}
 Finally for the last part of the sum we obtain
 \begin{align*}
       &\;\sum_{\ell = n + 1}^{2n}\sum_{k=0}^{\lfloor\frac{\ell}{2}\rfloor} \sum_{k'=0}^{\lfloor\frac{\ell}{2}\rfloor} n(k,\ell) \cdot n(k',\ell) \cdot \binom{3n}{\ell}^{-1}\\
  \leq &\;\sum_{\ell = n + 1}^{2n}\sum_{k=0}^{\lfloor\frac{\ell}{2}\rfloor} \sum_{k'=0}^{\lfloor\frac{\ell}{2}\rfloor} 2\binom{n}{k} \cdot 2\binom{n}{k'} \cdot \binom{3n}{\ell}^{-1}\\
  \stackrel{(\ref{eq:binom-mult})}{\leq} &\;\sum_{\ell = n + 1}^{2n}\sum_{k=0}^{\lfloor\frac{\ell}{2}\rfloor} \sum_{k'=0}^{\lfloor\frac{\ell}{2}\rfloor} 4\binom{2n}{k+k'} \cdot \binom{3n}{\ell}^{-1}\\
  \leq &\;\sum_{\ell = n + 1}^{2n}\sum_{k=0}^{\ell} 4\cdot\frac{\ell}{2} \cdot \binom{2n}{k} \cdot \binom{3n}{\ell}^{-1}\\
  \leq &\;\sum_{\ell = n + 1}^{2n} 2\ell \cdot 2^{2n} \cdot \binom{3n}{\ell}^{-1}\\
  \leq &\;4n^{2}\cdot2^{2n}\cdot\binom{3n}{n}^{-1}\\
  \stackrel{(\ref{eq:binom-3n-n})}{\leq} &\;4n^{2}\cdot2^{2n}\cdot\frac{2\sqrt{n}}{6^{n}}\\
  =    &\;8n^{2}\sqrt{n}\cdot\left(\frac{2}{3}\right)^{n}\\
  =    &\;\mathcal{O}\left(\frac{\log^{2}n}{n}\right).
 \end{align*}
\end{proof}

Now we have shown that~$B(G_n, \sigma)$ is with high probability odd, we also need to show that the graph~$B(G_n, \sigma)$ is with high probability rigid. Let us first observe that the automorphism group of~$G_n$ is the same as the automorphism group of the cycle~$C_{2n}$.

\begin{lemma}
 \label{le-aut}
 For $n \geq 4$, the automorphism group of the graph $G_n$ is \[\Aut(G_n) = D_{2n} = \langle(1,\dots,2n), (1,2n)\dots(n,n+1)\rangle\] the dihedral group of order~$2n$.
\end{lemma}

\begin{proof}
 Note that for~$n\geq4$ every 4-cycle of~$G_n$ contains at least one diagonal and consequently at least two diagonals. 
 We conclude that every diagonal of~$G_n$ is contained in two 4-cycles and that every non-diagonal is contained in only one 4-cycle.
 This implies that the set of diagonals is invariant under the automorphism group which implies that the automorphism group of~$G_n$ is the same as the automorphism group of the graph obtained by removing the diagonals, the cycle~$C_{2n}$.
\end{proof}

\begin{lemma}\label{lem:nr:of:involutions}
 Let $M_n$ be the set of those permutations of the set~$\{1,\ldots,n\}$ that are involutions. In other words, we set~$M_n := \{\varphi \in S_n\mid \varphi^{2} = 1\}$. Then $|M_n| \leq 2^{n-1} \cdot (n-1)!!$, where~$x!! = \prod_{i = 0}^{\lceil x/2\rceil -1}  (x-2i) = x(x-2)(x-4)\cdots$ is the double factorial.
\end{lemma}

\begin{proof}
 Let $\mathcal{I} \subseteq [n]$ be a set of even size $k$.
 Then there are at most $(k-1)!!$ elements in $M_n$ that move exactly those points in $\mathcal{I}$.
\end{proof}

\begin{lemma}
 \label{la:prob-rigid}
 The probability that $B(G_n, \sigma)$ is not rigid is in $\mathcal{O}(\frac{1}{n})$.
\end{lemma}
\begin{proof}
 Let $\sigma$ be a random permutation of $E(G_n)$.
 In order  to show the lemma we partition the set of permutations of the vertex set into three subsets and prove for each of them separately that it almost surely contains no automorphism.
 First note that vertices in~$V_B$ have degree~$3$ while vertices in~$W_B$ have degree four. Thus the sets~$V_B$ and~$W_B$ must be stabilized by every automorphism.

 \begin{claim}
 With probability $\mathcal{O}(\frac{1}{n})$, there is a non-trivial automorphism stabilizing each of the sets~$V(G_n) \times \{b\}$ for~$b\in\{0,1\}$.
 \vspace*{5pt}\\
 Let $\psi$ be such an automorphism and for~$b\in\{0,1\}$ denote by $\psi_b$ the restriction of $\psi$ to $V(G_n) \times \{b\}$.
 By the construction of $B(G_n, \sigma)$, each map $\psi_b$ has to be an automorphism of $G_n$ and $\psi(e) = \psi(\{v,w\}) = \{\psi_0(v), \psi_0(w)\}$ must hold for every edge $e = \{v,w\}\in E(G)$.
 For an automorphism~$\iota\in\Aut(G_n)$ let~$\iota^E$ be the permutation induced on the edges of~$G_n$, that is
 $\iota^E$ is the permutation on the edges of $G_n$ mapping $e = \{v,w\}$ to $\{\iota(v), \iota(w)\}$.
 Then $\psi_0^E(e) = \psi(e)$ for every edge $e \in E$.
 Furthermore $\psi_1^E(\sigma(e)) = \sigma(\psi(e))$.
 Together this means $\psi_1^E = {\sigma}\circ (\psi_0^E)\circ {\sigma}^{-1}$.
 
 We will now bound the probability for this equation to hold for a random~$\sigma$. We will use the following group theoretic fact: The number of permutations~$\sigma$ satisfying the equation is equal to the size of the centralizer~$C_{S_{3n}}(\psi_1^E)$ of~$\psi_1^E$ or~$0$ if no~$\sigma$ satisfies the equation. (Recall that this centralizer is the set of those permutations~$\sigma'$ with~$\sigma' \circ \psi_1^E \circ (\sigma')^{-1} = \psi_1^E$.)
 
 We will also use the fact that given a permutation with cycles length~$a_1,\ldots,a_t$ occurring with multiplicities~$n_1,\ldots,n_t$ the size of the centralizer is~$\prod_{i = 1}^t a_i^{n_i} (n_i)!$  (see~\cite[Subsection 2.2.3]{MR2129747} for basic information on centralizers).
 
 Let~$\Gamma\coloneqq \{\iota^E  \mid \iota\in \Aut(G) \} $ be the group of those permutations of the edges that are induced by automorphisms of~$G$.
 By Lemma \ref{le-aut}~$\Gamma$ is a dihedral group (in a non-standard action).
 Consider a rotation~$r$ of order~$t$. As a permutation of the edges, such a rotation has~$2n/t$ cycles of length~$t$ (on the non-diagonals) and either~$n/t$ cycles of order~$t$ if~$t$ is odd or~$2n/t$ cycles of order~$t/2$ if~$t$ is even (the diagonals).
 The centralizer has thus size~$|C_{S_{3n}}(r)| = t^{2n/t} (2n/t)! \cdot t^{n/t} (n/t)!$ or~$|C_{S_{3n}}(r)| = t^{2n/t} (2n/t)! \cdot (t/2)^{2n/t} (2n/t)!$.
 The second function is larger than the first (alternatively the argument given below for the second function directly translates to the first function).
 For the second function we get
 \[t^{2n/t} (2n/t)! \cdot (t/2)^{2n/t} (2n/t)! \leq \left(t^{1/t}\right)^{2n}n! \cdot \left((t/2)^{2/t}\right)^{n}n! \stackrel{(\ref{eq:sqrt})}{\leq} 2^{2n}n! \cdot 2^{n}n! \leq 8^{n}(2n)!.\]
 Next consider a reflection $s$.
 As a permutation of the edges, a reflection either consists of two cycles of length one and $(3n-2)/2$ two-cycles in case $n$ is even or $1$ one-cycle and $(3n-1)/2$ two-cycles or $3$ one-cycles and $(3n-3)/2$ two-cycles in case $n$ is odd.
 Thus the size of centralizer can be upper bounded by
 \[6 \cdot 2^{3n/2}\left(\frac{3n}{2}\right)! \leq 8^{n}(2n)!.\]
 Thus the probability that~$\psi_1^E = {\sigma}\circ (\psi_0^E)\circ {\sigma}^{-1}$ is at most~$\frac{8^{n} \cdot (2n)!}{(3n)!}$. Since~$|G| = 4n$, 
 the probability that this equation is fulfilled for some pair of automorphisms, that is $\Gamma \cap {\sigma}\Gamma{\sigma}^{-1} \neq \{1\}$, is at most $\frac{16n^{2}\cdot 8^{n} \cdot (2n)!}{(3n)!} \in \mathcal{O}(\frac{1}{n})$.
 \uend
 \end{claim}
 
 \begin{claim}[2]
 There is an automorphism mapping $V(G_n) \times \{0\}$ to~$V(G_n) \times \{1\}$ and vice versa with probability $\mathcal{O}(\frac{1}{n})$.
 \vspace*{5pt}\\
 Let $\psi$ be such an automorphism.
 Then $\psi^{2}$ stabilizes $V(G_n) \times \{b\}$, $b=0,1$, so using the first claim we can assume $\psi = \psi^{-1}$.
 Let $\psi_{b \rightarrow (1-b)}$ be the permutation on $V(G_n)$ that satisfies $\psi_{b \rightarrow (1-b)}(v) = v'$ whenever~$\psi(v,b)= (v',1-b)$.
 Then $\psi_{b \rightarrow (1-b)}$ is an automorphism of $G_n$ since for vertices~$u,v$ adjacent in~$G_n$ the vertices~$(u,b)$ and~$(v,b)$ have a common neighbor in~$B(G_n,\sigma)$ so their images~$\psi(u,b)$ and~$\psi(v,b)$ have a common neighbor in~$B(G_n,\sigma)$. Moreover $\psi_{0 \rightarrow 1} = \psi_{1 \rightarrow 0}^{-1}$.
 As before let~$\psi_{b \rightarrow (1-b)}^E$ be the permutation on the edges of $G_n$ mapping $e = \{v,w\}$ to $\{\psi_{b \rightarrow (1-b)}(v), \psi_{b \rightarrow (1-b)}(w)\}$.
 Then $\sigma(\psi(e)) = \{\psi_{0 \rightarrow 1}(v), \psi_{0 \rightarrow 1}(w)\} = \psi_{0 \rightarrow 1}^E(e)$ and $\psi(\sigma^{-1}(e)) = \psi_{1 \rightarrow 0}^E(e)$.
 Consequently we have $\sigma^{-1}(\psi_{0 \rightarrow 1}^E(e)) = \psi(e) = \psi_{1 \rightarrow 0}^E(\sigma(e))$ and $\sigma \cdot \psi_{1 \rightarrow 0}^E = (\sigma \cdot \psi_{1 \rightarrow 0}^E)^{-1}$.
 Note that for a random choice of~$\sigma$ the permutation~$\sigma \cdot \psi_{1 \rightarrow 0}^E$ is a random permutation drawn from~$S_{3n}$. However, the probability that such a random permutation is an involution is at most~$\frac{|M_{3m}|}{(3n)!}$.
 Overall, the probability that $\sigma \cdot \psi_{1 \rightarrow 0}^E = (\sigma \cdot \psi_{1 \rightarrow 0}^E)^{-1}$ is satisfied for some automorphism~$\psi_{1 \rightarrow 0}^{E}\in \Aut(G_n)$ is at most $\frac{4n \cdot |M_{3n}|}{(3n)!}$ which is in~$\mathcal{O}(\frac{1}{n})$ due to Lemma~\ref{lem:nr:of:involutions}.
 \uend
 \end{claim}
 
 \begin{claim}[3]
 There is an automorphism of other type (i.e, not respecting the partition~$\{V(G_n) \times \{0\}, V(G_n) \times \{1\}\}$) with probability $\mathcal{O}(\frac{1}{n})$.
 \vspace*{5pt}\\
 Let~$D$ be the graph obtained form a~$6$-cycle by adding a diagonal. Note first that every vertex~$v$ in~$G_n$ is contained in a subgraph of~$G_n$ isomorphic to~$D$ such that~$v$ is one of the two vertices of degree three in~$D$.
 
 This implies that every vertex~$(v,b)$ of~$B(G,\sigma)$ is contained in a subgraph isomorphic to~$D^{(2)}$, the graph obtained from~$D$ by subdividing every edge, such that~$(v,b)$ is one of the two vertices of degree three.
 
 Consequently, for every pair of vertices~$(v_0,0)$ and~$(v_1,1)$ there are two subgraphs~$D_0^{(2)}$ and~$D_1^{(2)}$ of~$B(G,\sigma)$ each isomorphic to~$D^{(2)}$ such that~$(v_i,i)$ is a vertex of degree~$3$ in~$D^{(2)}_i$ and such that~$D_i^{(2)}$ only contains vertices from~$V(G_n) \times \{i\} \cup E(G_n)$. In particular, the graphs~$D_0^{(2)}$ and~$D_1^{(2)}$ are vertex disjoint when it comes to vertices not being in~$E(G_n)$.
 
 We argue now that it is improbable that there are vertices~$v,u$ adjacent in~$G_n$ such that for~$(v,0)$ and~$(u,0)$ we can find such subgraphs~$D_0^{(2)}$ and~$D_1^{(2)}$.
 More precisely, let~$A(v,u)$ be the event that in~$G(B,\sigma)$ there are two graphs~$D_v^{(2)}$ and~$D_u^{(2)}$ each isomorphic to~$D^{(2)}$ such that~$V(D_0^{(2)})\setminus E(G_n)$ and~$V(D_1^{(2)}) \setminus E(G_n)$ are disjoint and such that~$v$ and~$u$ are vertices of degree~$3$ in~$D_v^{(2)}$ and~$D_u^{(2)}$, respectively. We argue that the probability of~$A(v,u)$ is at most~$\mathcal{O}(\frac{1}{n^2})$ for a given pair~$v,u$ of adjacent vertices of~$G_n$.
 
 Let~$E(v)$ be the event that~$(v,0)$ is contained in an~$8$-cycle of which a second neighbor of~$v$ lies in~$V(G_n) \times \{1\}$. Suppose that the event~$A(v,u)$ occurs. Then~$E(v)$ occurs due to the following argument. The vertex~$(v,0)$ is contained in two~$8$-cycles which share exactly one second neighbor of~$(v,0)$. However another second neighbor, namely~$(u,0)$, cannot be part of either of the cycles since the graphs~$D_0^{(2)}$ and~$D_1^{(2)}$ are disjoint when it comes to vertices in~$V(G_n) \times \{0\}$.
 
 We argue now that the probability of~$E(v)$ is in~$O(\frac{1}{n})$. 
 To see this, we consider several cases depending on how the vertices are distributed onto the two parts~$V(G_n) \times \{0\}$ and~$V(G_n) \times \{1\}$. Up to symmetry there are four options. The cycle has the form~$(0,0,1,1)$, $(0,1,0,1)$, $(0,1,1,0)$ or~$(0,1,1,1)$. (Here the notation describes the projection onto the second component of every second vertex on the cycle.)
 
 For the first case~$(0,0,1,1)$, it must be that~$v$ and a second neighbor of~$v$ both have second neighbors in Part~$1$ that are second neighbors. The probability for that is in~$\mathcal{O}(\frac{1}{n})$. 
 
 For the second case~$(0,1,0,1)$, it must be that two second neighbors~$x,x'$ of~$v$ in Part~$1$ have another common second neighbor $y$ in Part 0. This can happen in two ways. Either~$x$ and~$x'$ are second neighbors, the probability of which is~$\mathcal{O}(\frac{1}{n})$ or they are pairs with edges incident to the same vertex~$y$, the probability of which is also~$\mathcal{O}(\frac{1}{n})$.
 
 The third case is similar to the first.
 For the fourth case, the probability that~$v$ has two second neighbors that have a common second neighbor in Part~$1$ is~$\mathcal{O}(\frac{1}{n})$.  This proves that the probability of~$E(v)$ is in~$O(\frac{1}{n})$.
 
 To finish the proof we observe that the probability of~$E(u)$ under the condition that~$E(v)$ is also in~$\mathcal{O}(\frac{1}{n})$. Indeed, assuming one of the different possibilities in which~$E(v)$ can occur, the argument remains valid to show that~$E(u)$ has a probability of~$\mathcal{O}(\frac{1}{n})$.
 
 Since~$A(u,v)$ can only happen if both~$E(v)$ and~$E(u)$ happen we conclude that~$A(u,v)$ happens with probability~$\mathcal{O}(\frac{1}{n^2})$. Since~$G_n$ has only a linear number of edges, the probability that there exist adjacent~$u,v$ such that~$A(u,v)$ is in~$\mathcal{O}(\frac{1}{n})$.
 
 Finally, if~$\psi$ is an automorphism that does not respect the parts~$\{V(G_n) \times \{0\}, V(G_n) \times \{1\}\}$ then there is a pair of adjacent vertices~$u,v$ from the same part mapped to vertices in different parts. However, this implies the event~$A(u,v)$.
 \uend
 \end{claim}
\end{proof}

\begin{lemma}
 \label{la:prob-distance-two-twins}
 The probability that there are distinct vertices $w_1,w_2 \in W_B$ with $N_{B(G_n,\sigma)}^{2}(w_1) = N_{B(G_n,\sigma)}^{2}(w_2)$ is $\mathcal{O}(\frac{1}{n})$.
\end{lemma}

\begin{proof}
 Let $e_1,e_2 \in W_B = E(G_n)$ be distinct edges with $N_{B(G_n,\sigma)}^{2}(e_1) = N_{B(G_n,\sigma)}^{2}(e_2)$.
 For an edge $e \in E(G_n)$ we define the set $N(e) = \{e' \in E(G_n) \mid e \neq e' \wedge e \cap e' \neq \emptyset\}$.
 Observe that $N_{B(G_n,\sigma)}^{2}(e) = N(e) \cup \sigma^{-1}(N(\sigma(e)))$.
 We now estimate the probability that $N(e_1) \subseteq N(e_2) \cup \sigma^{-1}(N(\sigma(e_2)))$.
 
 For every two distinct edges $e, e' \in E(G_n)$ we have that $|N(e)| = 4$ and $|N(e) \cap N(e')| \leq 2$.
 Furthermore there are at most $2 \cdot |E(G_n)|$ many pairs of distinct edges $e,e' \in E(G_n)$ with $|N(e) \cap N(e')| = 2$.
 If $e$ is a diagonal edge then $e'$ has to be one of the two neighboring diagonal edges and otherwise $e'$ is the edge opposite to $e$.
 
 We now distinguish two cases.
 If $|N(e_1) \cap N(e_2)| = 2$ then $|N(e_1) \cap \sigma^{-1}(N(\sigma(e_2)))| \geq 2$.
 In other words $|N(e_1) \setminus N(e_2)| = 2$ and $\sigma(N(e_1)\setminus N(e_2)) \subseteq N(\sigma(e_2))$.
 The probability for this to happen for fixed edges $e_1,e_2 \in E(G_n)$ is $\frac{12}{(|E(G_n)|-1)\cdot(|E(G_n)|-2)}$.
 
 Otherwise $|N(e_1) \cap N(e_2)| \leq 1$.
 So $|N(e_1) \setminus N(e_2)| \geq 3$ and $\sigma(N(e_1)\setminus N(e_2)) \subseteq N(\sigma(e_2))$.
 The probability here is upper bounded by $\frac{24}{(|E(G_n)|-1)\cdot(|E(G_n)|-2)\cdot(|E(G_n)|-3)}$.
 
 Summing over all pairs of distinct edges $e_1,e_2 \in E(G_n)$ and using that the first case only occurs a linear number of times we obtain the desired bound.
\end{proof}

\begin{proof}[Proof of Theorem~\ref{thm:construction:is:rigid}]
 The theorem follows directly from Lemmas \ref{la:rigidity-cfi-graph}, \ref{la:prob-odd}, \ref{la:prob-rigid} and \ref{la:prob-distance-two-twins}.
\end{proof}

\end{appendix}

\end{document}